\RequirePackage{fix-cm}
\documentclass[smallextended]{svjour3}
\smartqed
\renewcommand{\makeheadbox}{} 
\usepackage[T1]{fontenc}
\usepackage{xurl}

\usepackage{times}
\usepackage{amssymb}
\usepackage{amsmath,amsfonts,amscd,upref,amstext,mathtools,comment}
\allowdisplaybreaks 
\usepackage{tikz}
\usetikzlibrary{trees}
\usepackage{graphicx}
\usetikzlibrary{positioning, quotes}
\usepackage{caption}
\usepackage{subcaption}
\usepackage{cancel}
\usepackage{tikz}
\usepackage{comment}
\usepackage{paralist}
\usepackage[utf8]{inputenc}
\usepackage{newunicodechar}
\usepackage{todonotes}

\makeatletter
\let\c@corollary\c@theorem

\let\p@corollary\p@theorem

\let\c@lemma\c@theorem

\let\p@lemma\p@theorem
\makeatother

\usepackage[colorlinks=true,urlcolor=blue, citecolor=red,linkcolor=blue,linktocpage,pdfpagelabels, bookmarksnumbered,bookmarksopen]{hyperref}
\usepackage[hyperpageref]{backref}
\newcommand{\doi}[1]{\url{https://doi.org/#1}}
\allowdisplaybreaks

\newcommand{\eps}{\varepsilon}

\newtheorem{fact}[theorem]{Fact}

\usepackage{xcolor}

\begin{document}

\title{Efficiently Coloring the Intersection of  a General Matroid and Combinatorial Matroids\thanks{A preliminary version of this paper appeared in IPCO 2026 \cite{arndt2026efficiently}.}}
\titlerunning{Coloring a General Matroid and Combinatorial Matroids}
\author{Stephen~Arndt \and Benjamin~Moseley \and Kirk~Pruhs \and Michael~Zlatin}
\authorrunning{Stephen Arndt et al.}

\institute{Stephen Arndt and Benjamin Moseley \at
Carnegie Mellon University, Pittsburgh, PA, USA\\
\email{\{sarndt,moseleyb\}@andrew.cmu.edu}\\
Corresponding author: Stephen Arndt (\href{mailto:sarndt@andrew.cmu.edu}{sarndt@andrew.cmu.edu})\\
Stephen Arndt: \href{https://orcid.org/0009-0008-2847-0721}{ORCID 0009-0008-2847-0721}\\
Benjamin Moseley: \href{https://orcid.org/0000-0001-8162-017X}{ORCID 0000-0001-8162-017X}
\and
Kirk Pruhs \at University of Pittsburgh, Pittsburgh, PA, USA\\
\email{kirk@cs.pitt.edu}\\
\href{https://orcid.org/0000-0001-5680-1753}{ORCID 0000-0001-5680-1753}
\and
Michael Zlatin \at Pomona College, Claremont, CA, USA\\
\email{michael.zlatin@pomona.edu}\\
\href{https://orcid.org/0000-0003-1773-1152}{ORCID 0000-0003-1773-1152}}

\date{}
\maketitle
\begin{abstract}
    This paper shows a polynomial-time algorithm that, given
  a general matroid $M_1$ 
   and $k-1$ partition matroids $ M_2, \ldots, M_k$,
  produces a coloring of the intersection  $M = \cap_{i=1}^k M_i$ 
  using at most $1+\sum_{i=1}^k \left(\chi(M_i) -1\right)$ colors.  This is the first polynomial-time $O(k)$-approximation algorithm 
  for  matroid intersection coloring where one of the matroids may be a general matroid.  Leveraging the fact that many of the standard combinatorial matroids reduce to partition matroids 
  at a loss of a factor of two in the chromatic number, this algorithm also yields a polynomial-time $O(k)$-approximation algorithm for matroid intersection coloring in the case where each of the matroids  $ M_2, \ldots, M_k$ are one of these standard combinatorial types. Even when $k = 2$, the best-known approximation ratio prior to the conference version of this work was $O(\log n)$ via a reduction to Set Cover.
\end{abstract}

\section{Introduction}
Some of the most generally applicable tools in an algorithmist's toolkit are efficient algorithms for various standard matroid-related optimization problems.
In particular, a reasonable number of natural optimization problems can be viewed
as set cover problems, in which the constraints can be partitioned into parts
that induce a matroid~\cite{part_decomp_1,schrijver_book}. 
Such problems can be viewed as matroid intersection coloring problems, defined 
as follows:

\begin{itemize}
    \item The
input is a collection of matroids $M_1 = (X, \mathcal{I}_1), \ldots,  M_k = (X, \mathcal{I}_k)$ on common ground set $X$.
The intersection $M = \cap_{i=1}^k M_i$ is defined
to be the set system $\left(X , \cap_{i=1}^k \mathcal{I}_i\right)$. 
\item A feasible solution is a coloring of the elements of  $X$ such that, for every color, the set of elements assigned that color is independent in 
$M$, i.e. independent in each matroid $M_i$ for $i \in [k]$. 
\item The objective is to \emph{minimize the number of colors} used. For any set system $M$, we use $\chi(M)$ to denote
the chromatic number of $M$, that is the minimum number of colors necessary for the existence
of a feasible coloring. 
\end{itemize}

There is a significant literature within the combinatorics/mathematics community
that establishes existential results related in some way
to matroid intersection coloring (and/or packing)~\cite{AharoniBergerGuoKotlar2025,BergerGuo2025,ab06,part_decomp_2,guo_24,galvin,lc_two_matroids,sbro}. 
However, most of the proofs of these results are nonconstructive, in that they are not readily 
adaptable to yield polynomial-time coloring algorithms. 
In particular, many of these results are established with topological fixed-point arguments that employ Sperner's lemma \cite{AharoniBergerGuoKotlar2025,BergerGuo2025,ab06}.
We now highlight some of these existential results that
are most relevant for this paper.

\begin{theorem} \cite{AharoniBergerGuoKotlar2025}
\label{thm:nonconstructive1}
For  $k$ general matroids $M_1, \ldots, M_k$,  it is the case that
 $\chi\left(\cap_{i=1}^k M_i\right) \le (2k-1) \max_{i=1}^k \chi(M_i)$.
\end{theorem}

For the intersection of two general matroids ($k=2$), this can be improved to the following.

\begin{theorem} \cite{BergerGuo2025}
\label{thm:nonconstructive2}
For  two general matroids $M_1$ and $ M_2$ it is the case that
 $\chi(M_1 \cap M_2) \le \chi(M_1) + \chi(M_2)$.
\end{theorem} 

The proofs of Theorems \ref{thm:nonconstructive1} and \ref{thm:nonconstructive2} both hinge on topological fixed-point arguments, deriving from an influential earlier paper 
by Aharoni and Berger \cite{ab06} that shows that $\chi(M_1 \cap M_2) \leq 2\max\{\chi(M_1), \chi(M_2)\}$.
These theorems show that in some sense 
the chromatic number of the intersection of $k$ matroids grows at most linearly with $k$. 
Further, there is a lower bound that rules out  the possibility of many natural types of sublinear dependence.

\begin{theorem} \cite{AharoniBergerGuoKotlar2025}
    \label{thm:oldpartitionlowerbound}
    There are infinitely many natural numbers $k$, such that there are $k+1$ partition
matroids $M_1, \ldots, M_{k+1}$, each with chromatic number $k$, and where the intersection of these
matroids has chromatic number  $k^2$. 
\end{theorem} 

However, the algorithmic question of producing such colorings for matroid intersection
lags behind what is known existentially. We now summarize the
literature on polynomial-time algorithms for matroid intersection coloring. Note that for the duration of the paper, we allow partition matroids to have arbitrary capacity constraints on each part (sometimes called ``generalized partition matroids'' in the literature).

\begin{compactitem}
\item \textbf{A single matroid.} Edmonds' algorithm optimally colors the elements in polynomial time~\cite{edmonds1968matroid,schrijver_book}. 
        \item \textbf{Two partition matroids.}  
      This case reduces to edge-coloring a bipartite graph and is solvable in polynomial time~\cite{schrijver_book}.
    \item \textbf{Two strongly-base-orderable  matroids.}  
      Davies and McDiarmid give a polynomial-time optimal algorithm for any pair of strongly-base-orderable matroids assuming oracle access to the bijective map~\cite{davies_and_mcd}.  
      Partition, transversal, and gammoid matroids are all strongly-base-orderable; but most notably, graphic matroids are not all strongly-base-orderable~\cite{brualdi,schrijver_book}.
    \item \textbf{$k$ partition matroids.}  
      A simple greedy procedure colors the intersection using $1+\sum_{i=1}^{k}\bigl(\chi(M_i)-1\bigr)$  colors.
      
    \item \textbf{$k$ combinatorial matroids.}
     If each of the matroids are one of the standard combinatorial types (graphic, laminar, transversal, partition, or a gammoid), then
there is a polynomial-time 
algorithm that uses $1+\sum_{i=1}^k \left( 2\chi(M_i) -1 \right)$ colors \cite{part_decomp_1,part_decomp_2,part_decomp_gammoid},
as all of these types of combinatorial matroids
can be efficiently reduced to a partition matroid at a cost of increasing
the chromatic number by at most a factor of two \cite{part_decomp_1,part_decomp_2,part_decomp_gammoid}.
In contrast it is known that no such reduction is possible for some binary matroids~\cite{imposs_1,imposs_2}.
 \item \textbf{$k$ general matroids.} After submission of the conference version of this work, it was shown that there is a polynomial-time algorithm that uses $(k^2-k)\max_{i \in [k]}\{\chi(M_i)\}$ colors for $k \geq 2$ \cite{arndt2026approximation} via polyhedral methods. We will give more context for how our results compare to \cite{arndt2026approximation} in Section \ref{subsec:our_results}. 
 \item \textbf{Instance optimality hardness.} It is known that there is no polynomial-time algorithm to optimally color the intersection of two general matroids under the independence oracle model \cite{BercziSchwarcz2021}. Further, it is NP-hard to optimally color the intersection 
 of a graphic matroid and a partition matroid~\cite{gmpm_hard,horsch2024rainbow}, or to optimally color the intersection of three partition matroids~\cite{obszarski}, when these matroids are given by an explicit representation.
\end{compactitem}

\paragraph{Bridging the constructive gap.}
Prior to the conference version of this work, there were no polynomial-time approximation algorithms known for matroid intersection coloring  if even one of the matroids is not one of the standard combinatorial types (except via a trivial reduction to Set Cover). Our goal is to match the existential results using constructive proofs, which also yield polynomial-time algorithms to find the desired coloring. Often, exactly matching nonconstructive existential results with constructive existential results is a daunting task.  Due to this, it is common to either settle for a slightly weaker result, or require slightly stronger preconditions (c.f.~\cite{Annamalai}).

In this paper, we consider matroid intersection coloring problems where all but one of the matroids is a partition matroid (or alternatively, by applying the known reductions to partition matroids, if all but one of the matroids are one of the standard combinatorial types).  Besides being a logical next step towards general matroid intersection coloring,  our motivation to consider such instances arises from the observation that some well-known conjectures in the combinatorics literature can be stated as conjectures about coloring the intersection of a general matroid and a collection of partition matroids. 
In particular, these matroid intersection coloring problems are related to Rota's Basis Conjecture~\cite{rota_basis} and  to the Strong Coloring Conjecture~\cite{scc}.
There  is significant literature around these well-known conjectures that contain  existential results with nonconstructive proofs. This gives   targets to shoot for when developing constructive polynomial-time algorithms. 

Further, computing the maximum number of disjoint common bases between two general matroids is reducible to the special case where one matroid is general and the other is partition \cite{HarveyKiralylau2011}. Thus it seems possible that such a reduction may also exist for minimum coloring, and coloring the intersection of one general matroid and one partition matroid may capture the complexity of coloring the intersection of two general matroids.

\subsection{Our Results}\label{subsec:our_results}

We work in the standard model where the general matroid is accessed via a polynomial-time independence oracle and the partition matroids are given in their explicit representation. The main result of this paper is:

\begin{theorem}\label{thm:main}
There is a polynomial-time algorithm that, given a general matroid $M_1 = (X, \mathcal{I}_1)$ 
   and $k-1$ partition matroids $ M_2, \ldots, M_k$,
   will produce a coloring of the intersection  $M = \cap_{i=1}^k M_i$ 
  using at most $1+\sum_{i=1}^k \left(\chi(M_i) -1\right)$ colors.
\end{theorem}

The upper bound in Theorem \ref{thm:main}
matches the standard upper bound when
all of the matroids are partition matroids, 
improves by one on the existential upper bound 
for two general matroids in Theorem \ref{thm:nonconstructive2},
and improves by a factor of two 
on the existential upper bound 
for $k$ general matroids from Theorem \ref{thm:nonconstructive1}. Another immediate corollary of Theorem \ref{thm:main} is:

\begin{corollary}\label{cor:main}
There is a polynomial-time algorithm that, given a general matroid $M_1 = (X, \mathcal{I}_1)$ 
   and $k-1$  matroids $ M_2, \ldots, M_k$ that are each of the following types: graphic, laminar, transversal,  partition, or a gammoid, 
   will produce a coloring of the intersection  $M = \cap_{i=1}^k M_i$ that has worst-case approximation ratio $2k-1$. 
\end{corollary}

Corollary \ref{cor:main} follows from Theorem \ref{thm:main} by first  noting that the
chromatic number of the intersection of matroids is at least the maximum chromatic number of the individual matroids and then  applying the known efficient reductions of  these standard combinatorial matroids to partition matroids ~\cite{part_decomp_1,part_decomp_2,part_decomp_gammoid}.
Taking $k = 2$ in~Corollary \ref{cor:main} yields:

\begin{corollary}\label{cor:k=2}
There is a polynomial-time 3-approximation for matroid intersection coloring of a general matroid and a matroid which is either graphic, laminar, transversal,  partition, or a gammoid.
\end{corollary}

In the $k=2$ case, Corollary \ref{cor:k=2} gives the first polynomial-time $O(1)$-approximation algorithm for matroid intersection coloring where one of the matroids may be a general matroid. Prior to the conference version of this work, the best-known approximation for coloring a general and partition matroid was $O(\log n)$, through a reduction to Set Cover (where the sets are given through oracle access).

\medskip

\noindent \textbf{Comparison with subsequent work.} As mentioned previously, after submission of the conference version of this work, it was shown that there is a polynomial-time algorithm that produces a coloring of the intersection of $k$ general matroids using at most $(k^2-k)\max_{i \in [k]}\{\chi(M_i)\}$ colors for $k \geq 2$ \cite{arndt2026approximation}. We work in the more restricted setting of one general and $k-1$ combinatorial matroids, but produce a coloring using at most $(2k-1)\max_{i \in [k]}\{\chi(M_i)\}$ colors, so our results are incomparable. 

One notable exception is the $k=2$ case, where in the setting of one general and one combinatorial matroid, applying subsequent work yields a $2$-approximation while applying our work (Corollary \ref{cor:k=2}) yields a $3$-approximation. Similarly, in the setting of one general and one partition matroid, applying subsequent work yields a $2$-approximation ($2\max\{\chi(M_1), \chi(M_2)\}$ colors) while applying our work (Theorem \ref{thm:main}) also yields a $2$-approximation ($\chi(M_1) + \chi(M_2)-1$ colors). Our main additional contribution to the $k=2$ case is in the setting of one general and one partition matroid, where our additive upper bound $\chi(M_1 \cap M_2) \leq \chi(M_1)+\chi(M_2)-1$ actually matches (and improves by one) on the existential upper bound in Theorem \ref{thm:nonconstructive2}.

\medskip
\noindent \textbf{The difficulty of matroid intersection coloring.}
A natural algorithmic design approach is to, like the standard Edmonds' algorithm for coloring one matroid \cite{edmonds1968matroid,schrijver_book}, try
to design an algorithm where after $q$ iterations, it is the case that $q$ items are
feasibly colored. Thus, on each iteration, items are recolored to result in a feasible coloring
of one more element. Let us first consider the problem
of coloring a single matroid. The difficulty is that when an item is recolored, say item $x$ is recolored blue
for concreteness, then the feasibility constraints may require that some other
item, say $y$,  currently
colored blue, be recolored to some other color. In such a case, let us say that
the recoloring of $x$ to blue was fettered. 
Edmonds' algorithm for coloring one matroid addresses this issue by, on each iteration,
performing an augmenting sequence of recolorings, where each recoloring corrects the infeasibility created
by the prior fettered recoloring, and the last recoloring is unfettered.

The standard interpretation of Edmonds' algorithm is to find such an augmenting sequence
by identifying a \emph{shortest path} in a particular digraph, that we will call Edmonds' digraph~\cite{schrijver_book}.
However, when attempting to adapt this approach to coloring two (or more) matroids, 
recoloring $x$ to blue may require that \emph{two (or more)} elements currently colored blue be recolored.
This naturally leads to considering algorithms where the augmenting path in Edmonds' algorithm
is replaced by an augmenting tree of hyperedges of cardinality three, 
\'a la Haxell's hypergraph matching algorithm~\cite{BF01793010}. 
However, it is not at all clear how to efficiently find such an augmenting tree that would produce a feasible
recoloring of one more item.~\footnote{For hypergraph matching, this was complicated but doable~\cite{Annamalai},
essentially because the matroids involved were all partition matroids. }
Rather than surmount this issue, our algorithm that establishes Theorem \ref{thm:main}
circumvents the issue, which we now explain.

\medskip
\noindent \textbf{Generalized Edmonds' Algorithm.}
In Section~\ref{sec:prelim}, we explain one of the two main components in the design of
the algorithm that establishes Theorem \ref{thm:main}, namely a modest generalization of Edmonds' algorithm,
which we call Generalized Edmonds' Algorithm. We emphasize that Generalized Edmonds' Algorithm is essentially a known algorithm and the generalization follows directly from Lemma 13.1.11 in~\cite{FRANK20121875}. However, the definitions and algorithm description in Section \ref{sec:prelim} will be important for understanding our matroid intersection coloring algorithm in Section \ref{sec:alg}.

One issue with building on Edmonds' algorithm to color multiple matroids, 
is that it is easy to construct instances where augmenting along the  shortest path $P$ 
in Edmonds digraph for matroid $M_1$ will not produce a feasible coloring
with respect to the other matroids. Thus we need greater flexibility with respect to how matroid $M_1$ 
can be recolored. Upon closer examination of the 
standard proof of correctness for Edmonds' algorithm~\cite{schrijver_book},
one can observe that it is sufficient that $P$ is what we will call a color-chordless path,
which is roughly a path where there is no chord/shortcut between vertices that
involve the same color. A shortest path, as it has no chords, is trivially a color-chordless path. The Generalized Edmonds' Algorithm augments along an arbitrary color-chordless path 
on each iteration, giving us greater flexibility when trying to adapt the algorithm to 
coloring multiple matroids.

\medskip
\noindent \textbf{Our Matroid Intersection Coloring Algorithm.}
In Section~\ref{sec:alg}, we design and analyze the algorithm that establishes Theorem \ref{thm:main}.
At a high level, our algorithm finds a color-chordless path $P$ in the Edmonds digraph for
matroid $M_1$ that is suffix-feasible with respect to each of the partition matroids $M_2, \ldots, M_k$. 
The path $P$ is suffix-feasible with respect to a matroid $M_i$ if   iteratively augmenting along $P$ (as in Edmonds' algorithm) results in a sequence of recolorings where
each recoloring 
is unfettered in  $M_i$ \emph{at the time of the recoloring}. 
To accomplish this, the algorithm first constructs a subgraph $H$ of the Edmonds digraph $G$
for $M_1$, such that in $H$  every source-sink
path is a color-chordless path in $G$. The algorithm then finds a source-sink path $P$ within $H$,
that is suffix-feasible with respect to each of the partition matroids $M_2, \dots, M_k$, 
using an iterative greedy algorithm.

\medskip\noindent \textbf{Organization of paper.} In Section \ref{app:def}, we review some basic matroid definitions and facts. In Section \ref{sec:prelim}, we explain Generalized Edmonds' Algorithm. In Section \ref{sec:alg}, we state our matroid intersection coloring algorithm, and prove its correctness and runtime guarantee. In Section \ref{app:hardness}, we show that it is NP-hard to approximate the coloring of $k$ \textit{partition} matroids within a factor of $k^{1-\eps}$ for any $\eps > 0$ via a reduction from vertex coloring. In Section \ref{app:applications}, we discuss applications of our matroid intersection coloring algorithm to generalizations of Rota's Basis Conjecture \cite{rota_basis} and the Strong Coloring Conjecture \cite{scc}, and we extend our algorithm to the setting of list coloring.

\medskip \noindent
\textbf{Relation to the conference version.}
Compared with the IPCO version \cite{arndt2026efficiently},
this version adds the worked example in Section~\ref{app:generalizededmondsexample},
the approximation-hardness results in
Section~\ref{app:hardness},
and the applications and list-coloring extension in
Section~\ref{app:applications}.

\section{Preliminaries}\label{app:def}

We begin by reviewing some basic matroid definitions and facts.

\begin{definition}[Matroid]\label{def:matroid}
A \textbf{matroid} is a pair $M = (X, \mathcal{I})$ where $X$ is called the ground set, and $\mathcal{I} \subseteq 2^X$ are called the independent sets, with the following properties:

\begin{enumerate}
    \item \textbf{Empty Set:} $\emptyset \in \mathcal{I}$.
    \item \textbf{Subset Property:} If $J \in \mathcal{I}$ and $I \subseteq J$, then $I \in \mathcal{I}$.
    \item \textbf{Exchange Property:} If $I, J \in \mathcal{I}$ and $|J| > |I|$, then there exists an $x \in J \setminus I$ s.t. $I \cup \{x\} \in \mathcal{I}$.
\end{enumerate}
\end{definition}

\begin{definition}[Partition Matroid]\label{def:partition}
A \textbf{partition matroid} is a matroid $M = (X, \mathcal{I})$ defined by a partition $X = X_1 \cup X_2 \cup \dots \cup X_m$ with associated capacity constraints $d_1, d_2, \dots, d_m$, where $\mathcal{I} = \{I \subseteq X : |I \cap X_i| \leq d_i \ \forall i = 1, 2, \dots, m\}$. If capacity constraints are not given, they are typically assumed to be $1$.
\end{definition}

\begin{definition}[Rank Function]
For a matroid $M = (X, \mathcal{I})$, and for all $A \subseteq X$, the \textbf{rank function} $r(A)$ is the size of the largest independent set $I \subseteq A$.
\end{definition}

\begin{definition}[Laminar Matroid]\label{def:laminar}
A laminar matroid $M = (X, \mathcal{I})$ is a matroid associated with a laminar family $\mathcal{F}$ of subsets of $X$. A family of subsets is said to be laminar if for any $A, B \in \mathcal{F}$, $A \cap B = \emptyset$, $A \subseteq B$, or $B \subseteq A$. Further, each set $A \in \mathcal{F}$ is associated with a positive integer $b(A)$, which is called the \textit{capacity} of the set. A set $I \subseteq X$ is independent in $M$ if $|I \cap A| \leq b(A)$ for all $A \in \mathcal{F}$.
\end{definition}

\begin{definition}[Circuit]\label{def:circuit}
For a matroid $M = (X, \mathcal{I})$, a \textbf{circuit} $C \subseteq X$ of $M$ is an inclusion-minimal dependent set.
\end{definition}

\begin{fact}[Unique Circuits \cite{schrijver_book}]\label{fact:unique_circuit}
For a matroid $M = (X, \mathcal{I})$, a subset $A \subseteq X$ has $r(A) = |A|-1$ if and only if $A$ contains a unique circuit $C$.
\end{fact}

\begin{fact}[Matroid Circuit Axiom \cite{schrijver_book}]\label{fact:mca}
For a matroid $M = (X, \mathcal{I})$ and any two distinct circuits $C, C' \subseteq X$ with common element $x$, there exists a circuit $C'' \subseteq (C \cup C') - \{x\}$.
\end{fact}

\section{Generalized Edmonds' Algorithm for Coloring a Single Matroid}\label{sec:prelim}

This section gives a generalization of Edmonds' algorithm \cite{edmonds1968matroid,schrijver_book} for coloring the elements of a single arbitrary matroid optimally that follows directly from \cite{FRANK20121875}. In Section \ref{subsec:gea_desc}, we describe the algorithm. In Section \ref{subsec:gea_analysis}, we prove that the algorithm produces a feasible coloring that uses at most the desired number of colors. In Section \ref{app:generalizededmondsexample}, we show one iteration of the algorithm on a particular instance.

\subsection{Algorithm Description}\label{subsec:gea_desc}

As with the standard Edmonds' algorithm, the Generalized Edmonds' algorithm has the following input, output, and outer loop invariant.

\textbf{Input:} An arbitrary matroid $M$ and target number of colors $\alpha$.

\textbf{Output:} An $\alpha$-coloring of $M$ if one exists.

\textbf{Outer Loop Invariant:} After $q$ iterations, $q$ elements of the ground set are feasibly colored. A coloring is \textit{feasible} if its color classes are independent sets in the matroid.

\textbf{Computing a Minimum Coloring:} We take as given that if $\alpha \geq \chi(M)$, the algorithm will successfully produce an $\alpha$-coloring of all elements \cite{edmonds1968matroid,schrijver_book}. Thus we can compute a minimum coloring of $M$ (i.e. a $\chi(M)$-coloring) by using binary search to find the minimum $\alpha$ where the algorithm finds a feasible coloring.

\medskip 
\noindent \textbf{Description of the Generalized Edmonds' Algorithm:} Iteratively update the current coloring $c$ along an arbitrary color-chordless source-sink path $P$ in the current Edmonds digraph $G$. If all elements are colored, return the coloring. If no source-sink path exists and $q < n$, report failure.

The standard interpretation of Edmonds' algorithm~\cite{schrijver_book} is an instantiation of this algorithm in which the arbitrary color-chordless source-sink path is selected to be a shortest source-sink path. To finish the description, we need to define color-chordless source-sink path. Then we remind the reader of the definition of the Edmonds digraph $G$  and how a coloring is updated along a path (as these are the same in the standard description of Edmonds' algorithm~\cite{schrijver_book}).

\begin{definition}[Edmonds Digraph]\label{def:edmonds_digraph}
The definition of Edmonds digraph $G=(V, A)$ depends on the matroid $M=(X, \mathcal{I})$, the target number of colors $\alpha$, and the current coloring $c$.
Let $c : X \rightarrow \{0, 1, \dots, \alpha\}$ be the current coloring of elements, where $c(x) = 0$ if $x$ is uncolored. Let $S_1, S_2, \dots, S_{\alpha} \in \mathcal{I}$ be the color classes of $c$. Let $U \subseteq X$ be the set of elements not colored by $c$.
The vertex set $V = [\alpha] \cup X$ consists of a vertex for each element of the ground set $X$ of $M$ and one
vertex for each color. We call the vertices in $[\alpha]$ corresponding to the $\alpha$ available colors the \textbf{sources} and the uncolored elements in $U$ the \textbf{sinks.} The arc set $A= \cup_{i=1}^{\alpha} A_i$
is the union of $\alpha$ collections of arcs, where $A_i$ represents the feasible exchanges for color class $i$. In $A_i$ there is an arc $(i, x)$ to all elements $x \notin S_i$ with the property that $S_i \cup \{x\} \in \mathcal{I}$. Then, for all elements $x \notin S_i$ such that $S_i \cup \{x\} \notin \mathcal{I}$, $A_i$ contains the arc $(y, x)$ for all elements $y \in S_i$ with the property that $S_i - \{y\} \cup \{x\} \in \mathcal{I}$.
\end{definition}

\begin{definition}[Color-Chordless Path]\label{def:color-chordless_path}
Let $P = (x_1, x_2, \dots, x_{\ell})$ be a directed path in $G$. A \textbf{color-chord} of $P$ is an arc $(x_j, x_k) \in A$ between two elements $x_j, x_k \in X$ where $1 \leq j \leq k-2$, and $x_j$ is the same color as $x_{k-1}$ in the current coloring $c$ (note this is $x_{k-1}$ and not $x_k$). $P$ is \textbf{color-chordless} if it has no color-chords.
\end{definition}

\noindent
{\bf Definition of Updating Along $P$:} Let $P = (x_1, x_2, \dots, x_{\ell})$ be a source-sink path in $G$. The updated coloring $c \Delta P$ is computed by coloring $x_2$ with the color of $x_1$, and coloring $x_i$ with the color of $x_{i-1}$ in $c$ for all $3 \leq i \leq \ell$.

\subsection{Algorithm Analysis}\label{subsec:gea_analysis}

We take as given that a source-sink path always exists when $\alpha \geq \chi(M)$ \cite{edmonds1968matroid,schrijver_book}. Thus the algorithm successfully produces an $\alpha$-coloring of all elements when $\alpha \geq \chi(M)$, and a color-chordless source-sink path can be efficiently computed (for example, by computing a shortest source-sink path). We now show that upon updating the current coloring $c$ along an arbitrary color-chordless source-sink path $P$, the resulting coloring $c \Delta P$ is still feasible. This follows directly from Lemma 13.1.11 in \cite{FRANK20121875} but we include a full proof for completeness.

\begin{lemma}\label{lemma:color-chordless_path}
Let $c$ be a feasible coloring of $q$ elements in $M$. Let $G$ be the Edmonds digraph of $M, c$. Let $P = (x_1, x_2, \dots, x_{\ell} )$ be a color-chordless source-sink path in $G$. Then $c \Delta P$ is a feasible coloring of $q+1$ elements in $M$.
\end{lemma}

\begin{proof}
The proof structure is as follows. First, we assume for sake of contradiction that $c \Delta P$ is not feasible in $M$, and find the first element which is not feasibly recolored. Next, we study two circuits (see Definition \ref{def:circuit}) $C, C'$ in $M$ containing this element. Lastly, we appeal to the Matroid Circuit Axiom (Fact \ref{fact:mca}) to produce a third circuit $C''$ in $M$, which we use to establish the contradiction.

Note $c \Delta P$ is a coloring of $q+1$ elements, because $x_2, x_3, \dots, x_{\ell-1}$ remain colored and $x_\ell$ becomes colored. Assume for the sake of contradiction that $c \Delta P$ is not feasible in $M$. Recolor $x_2, x_3, \dots, x_{\ell}$ according to $P$ sequentially, and let $x_t$ be the first element for which recoloring $x_t$ makes the coloring infeasible. Suppose $x_t$ is recolored with color $j$.

Let $S_j \in \mathcal{I}$ be the $j$th color class in $c$, and $S_j' \in \mathcal{I}$ be the $j$th color class immediately prior to recoloring $x_{t-1}$. Note $t \geq 3$, because recoloring $x_2$ with color $x_1$ is a feasible recoloring by definition of $G$. Thus $S_j \cup \{x_t\} \notin \mathcal{I}$ because of the existence of arc $(x_{t-1}, x_t)$ in $G$, so there exists a circuit $C \subseteq S_j \cup \{x_t\}$ containing $x_t$. Further, $x_{t-1} \in C$, again because of the existence of arc $(x_{t-1}, x_t)$ in $G$ which implies $S_j - \{x_{t-1}\} \cup \{x_t\} \in \mathcal{I}$. Lastly, no elements in $C-\{x_t\}$ are recolored before $x_{t-1}$. This is because all elements $x \in C-\{x_t\}$ have $S_j - \{x\} \cup \{x_t\} \in \mathcal{I}$, and thus $(x, x_t) \in A$. Therefore if $x_h \in C - \{x_t\}$ for some $h \leq t-2$, then $(x_h, x_t) \in A$ would be a color-chord of $P$, a contradiction. Thus $C - \{x_t\} \subseteq S_j'$.

Because recoloring $x_t$ makes the coloring infeasible, we have $S_j' - \{x_{t-1}\} \cup \{x_t\} \notin \mathcal{I}$, so there exists a circuit $C' \subseteq S_j' - \{x_{t-1}\} \cup \{x_t\}$ containing $x_t$. Further, $C \neq C'$ because $C$ contains $x_{t-1}$ but $C'$ does not. Thus by the Matroid Circuit Axiom (Fact \ref{fact:mca}), there exists a circuit $C'' \subseteq (C \cup C') - \{x_t\} \subseteq S_j'$. But $S_j' \in \mathcal{I}$, a contradiction.
\end{proof}

\subsection{Coloring a Single Matroid on an Example}
\label{app:generalizededmondsexample}

Consider the following matroid $M$ in Figure \ref{fig:sec2_matroid} and current $2$-coloring of elements $c$ in Figure \ref{fig:sec2_init_coloring}. The construction of the Edmonds digraph $G$ is shown in Figure \ref{fig:sec2_G1G2G}. An example shortest source-sink $P$ and the resulting coloring $c \Delta P$ are shown in Figure \ref{fig:sec2_shortest_path}. An example color-chordless source-sink path $P$ and the resulting coloring $c \Delta P$ are also shown in Figure \ref{fig:sec2_shortest_path}. 

\begin{figure}[h]
\begin{center}
\begin{subfigure}[t]{.3\textwidth}
\includegraphics[width=\linewidth]{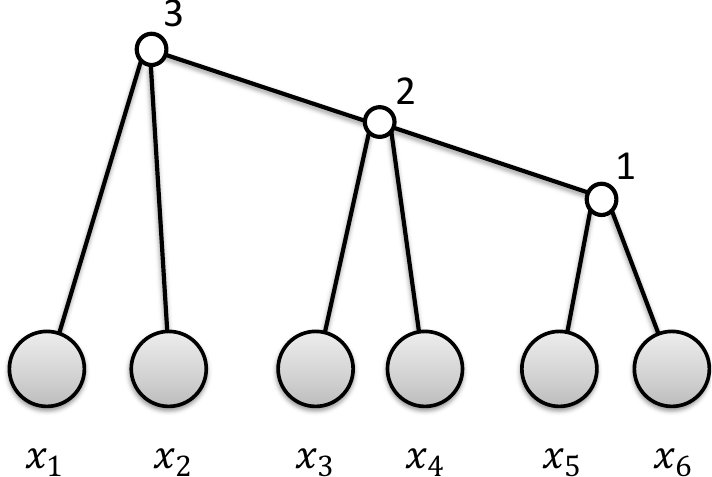}
\caption{}
\label{fig:sec2_matroid}
\end{subfigure}%
\qquad
\begin{subfigure}[t]{.3\linewidth}
\includegraphics[width=\linewidth]{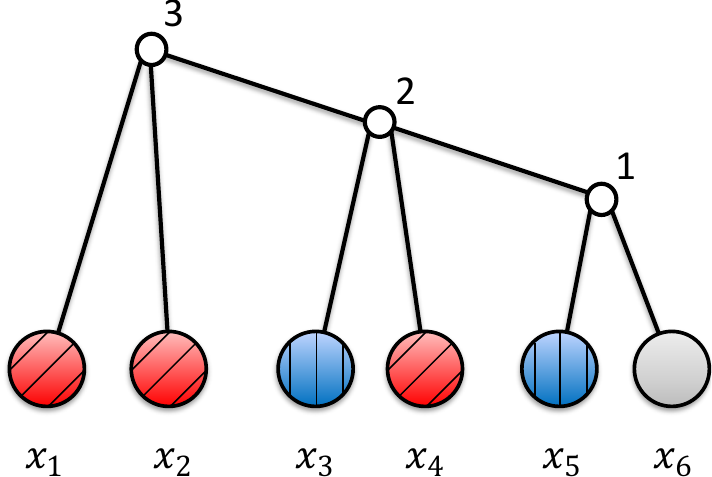}
\caption{}
\label{fig:sec2_init_coloring}
\end{subfigure}
\caption{(\ref{fig:sec2_matroid}) A laminar matroid  and (\ref{fig:sec2_init_coloring}) a corresponding partial coloring.  
}
\end{center}
\end{figure}

The example matroid in Figure \ref{fig:sec2_matroid} is a 2-colorable laminar matroid (see Definition \ref{def:laminar}). The ground set is $X = \{x_1, x_2, x_3, x_4, x_5, x_6\}$ with laminar family $\mathcal{F} = \{\{x_5, x_6\}, \{x_3, x_4, x_5, x_6\}, \{x_1, x_2, x_3, x_4, x_5, x_6\}\}$. The set $\{x_5, x_6\}$ has capacity $1$, $\{x_3, x_4, x_5, x_6\}$ has capacity $2$, and $\{x_1, x_2, x_3, x_4, x_5, x_6\}$ has capacity $3$.  Figure \ref{fig:sec2_init_coloring} shows a partial $2$-coloring of elements in the laminar matroid. Color $1$ is blue (vertical hatches), and color $2$ is red (diagonal hatches). We have $c(x_6) = 0, c(x_3) = c(x_5) = 1$, and $c(x_1) = c(x_2) = c(x_4) = 2$. $S_1 = \{x_3, x_5\}, S_2 = \{x_1, x_2, x_4\}$, and $U = \{x_6\}$. Note $S_1, S_2 \in \mathcal{I}$.

Next, we illustrate the construction of the Edmonds digraph. Figure \ref{fig:a1} shows edges in $A_1$. Note $S_1 \cup \{x_1\} = \{x_1, x_3, x_5\} \in \mathcal{I}$, so there is an arc $(1, x_1) \in A_1$. Figure \ref{fig:a2} shows the edges $A_2$. Note $S_2 \cup \{x_3\} = \{x_1, x_2, x_3, x_4\} \notin \mathcal{I}$ but $S_2 - \{x_4\} \cup \{x_3\} = \{x_1, x_2, x_3\} \in \mathcal{I}$, so there is an arc $(x_4, x_3) \in A_2$. Finally, Figure \ref{fig:edmondgraph} shows the entire graph $G$. Note $V(G) = \{1, 2\} \cup X$ and $A = A_1 \cup A_2$.

\begin{figure}
\begin{center}
\begin{subfigure}[t]{.27\textwidth}
\includegraphics[width=\textwidth]{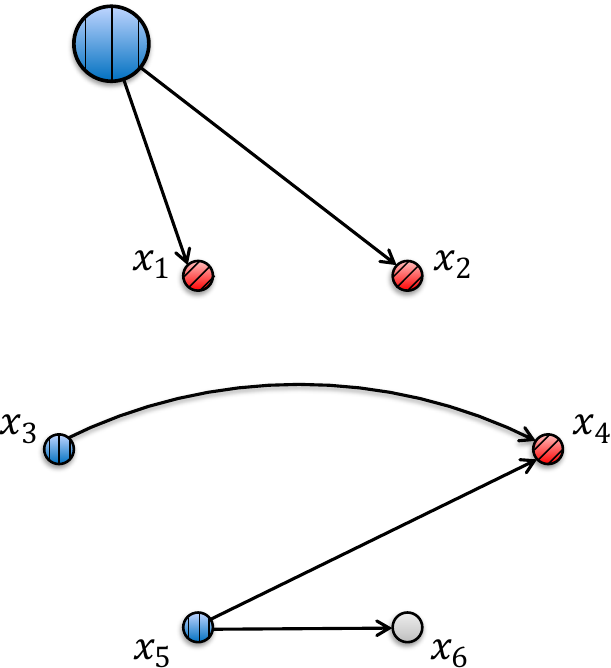}
\caption{}
\label{fig:a1}
\end{subfigure}
\begin{subfigure}[t]{.27\textwidth}
\includegraphics[width=\textwidth]{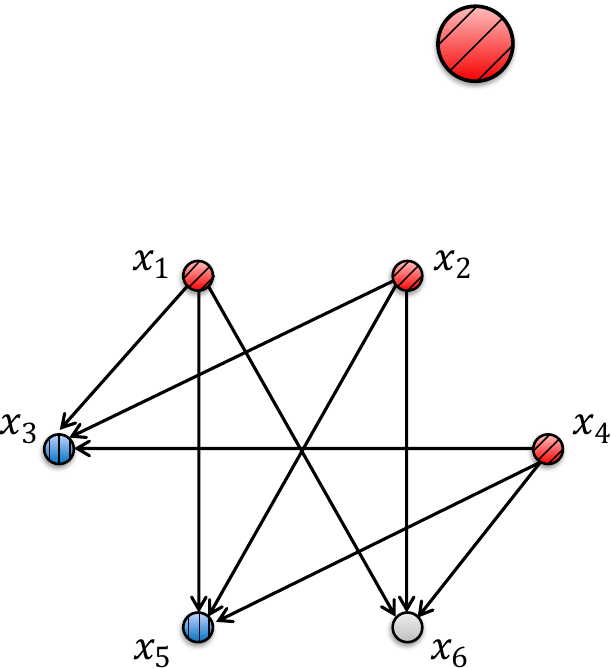}
\caption{}
\label{fig:a2}
\end{subfigure}
\begin{subfigure}[t]{.27\textwidth}
\includegraphics[width=\textwidth]{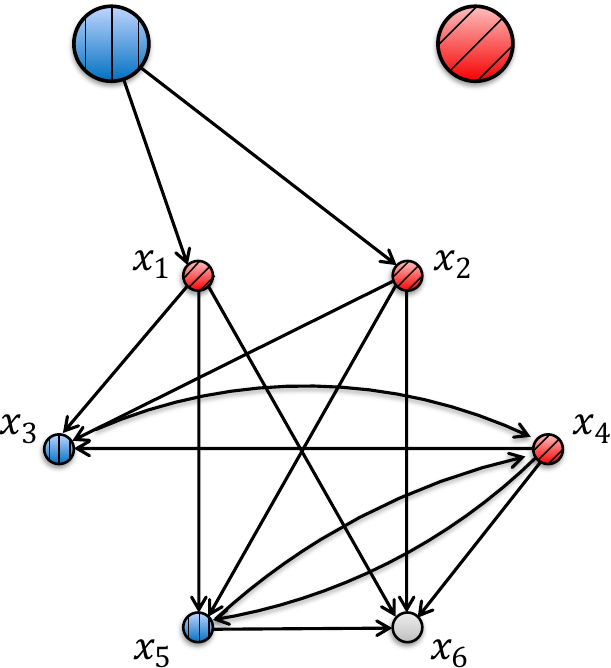}
\caption{}
\label{fig:edmondgraph}
\end{subfigure}
\end{center}
\caption{(\ref{fig:a1}) is the edges in $A_1$, (\ref{fig:a2}) is the edges in $A_2$, and (\ref{fig:edmondgraph}) is the Edmonds digraph.}
\label{fig:sec2_G1G2G}
\end{figure}

Generalized Edmonds' algorithm then identifies a color-chordless source-sink path. First we show a path Edmonds' algorithm could identify and then we show a path that is only possible in our generalization of Edmonds' algorithm. Figure \ref{fig:posiblepath} shows a possible \emph{shortest} source-sink path $P = (1, x_2, x_6)$ (a path which can be used in the original definition of Edmonds' algorithm). Next, Figure \ref{fig:updatecolor} shows the updated coloring $c \Delta P$ according to the path. Note $c \Delta P$ is a feasible coloring in $M$, because $\{x_2, x_3, x_5\}, \{x_1, x_4, x_6\} \in \mathcal{I}$.

\begin{figure}[h]
\begin{center}
\begin{subfigure}[t]{.22\textwidth}
\includegraphics[width=\textwidth]{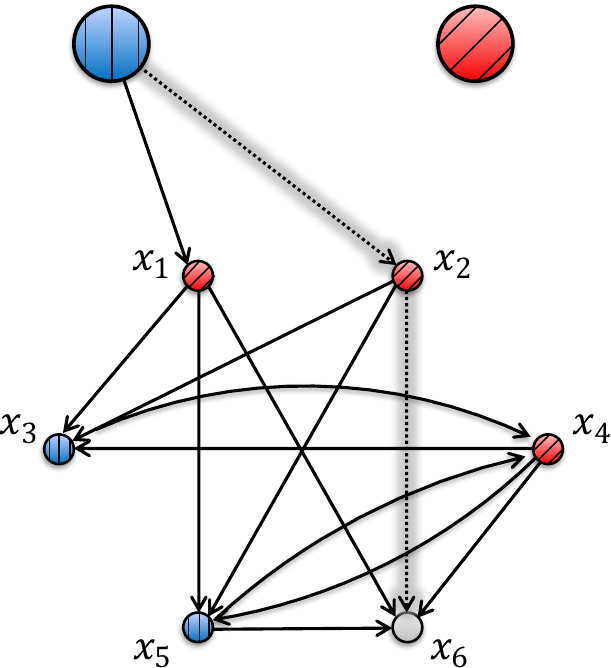}
\caption{}
\label{fig:posiblepath}
\end{subfigure}
\begin{subfigure}[t]{.22\textwidth}
\includegraphics[width=\textwidth]{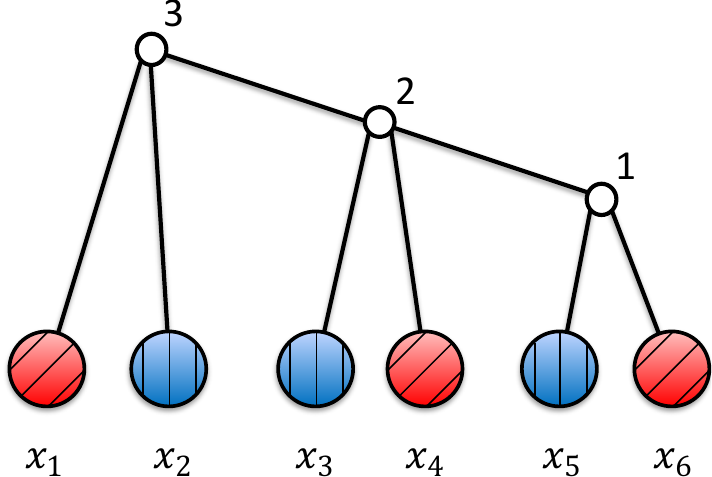}
\caption{}
\label{fig:updatecolor}
\end{subfigure}
\begin{subfigure}[t]{.22\textwidth}
\includegraphics[width=\textwidth]{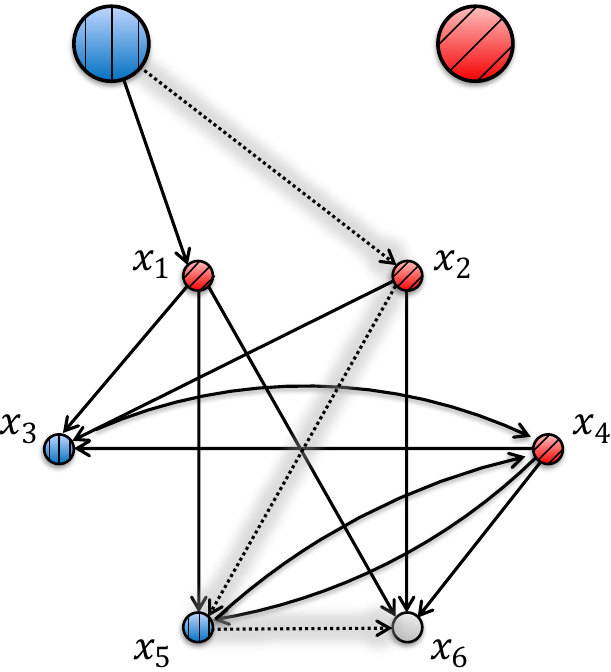}
\caption{}
\label{fig:genpath}
\end{subfigure}
\begin{subfigure}[t]{.22\textwidth}
\includegraphics[width=\textwidth]{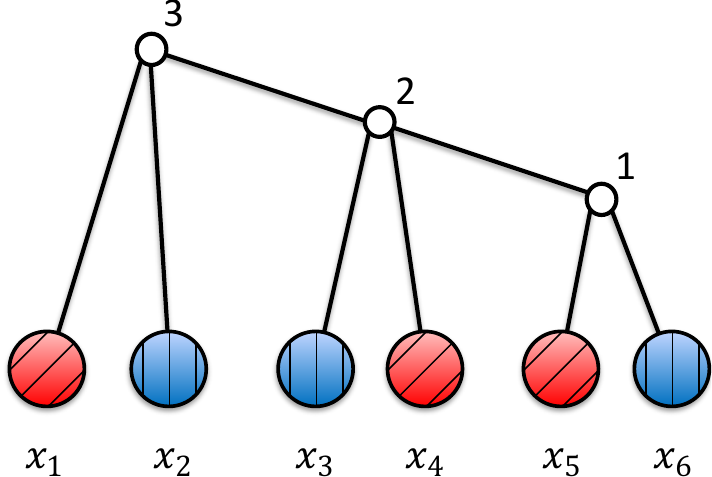}
\caption{}
\label{fig:genpathupdate}
\end{subfigure}
\end{center}
\caption{ (\ref{fig:posiblepath}) is a shortest source-sink path and (\ref{fig:updatecolor}) is the corresponding updated coloring. (\ref{fig:genpath}) is a color-chordless source-sink path and (\ref{fig:genpathupdate}) is the corresponding updated coloring.}
\label{fig:sec2_shortest_path}
\end{figure} 

Figure \ref{fig:genpath} shows a possible non-shortest color-chordless source-sink path $P = (1, x_2, x_5, x_6)$ and Figure \ref{fig:genpathupdate} shows the corresponding updated coloring $c\Delta P$. This is the kind of path that can be identified in our
generalization of Edmonds' algorithm. Note $c \Delta P$ is a feasible coloring in $M$, because $\{x_2, x_3, x_6\}, \{x_1, x_4, x_5\} \in \mathcal{I}$.

\section{Matroid Intersection Coloring Algorithm }\label{sec:alg}

In this section we prove our main result, Theorem \ref{thm:main}, that 
there is an efficient algorithm for coloring the intersection of a general matroid
and a collection of partition matroids. 
In Section \ref{subsec:intersectionalgorithmdescription}, we state the algorithm. In Section \ref{subsec:intersectionanalysis}, we prove that the algorithm 
produces a feasible coloring that uses at most the desired number of colors. In Section \ref{subsec:runtime}, we prove that the algorithm has runtime $O(n^3 T(M_1))$, where $n$ is the number of elements and $T(M_1)$ is the runtime of a single independence oracle call to the general matroid $M_1$. In Section \ref{app:intersectionexample}, we show one iteration of the algorithm on a particular instance.

\subsection{Algorithm Description }
\label{subsec:intersectionalgorithmdescription}

Our Matroid Intersection Coloring Algorithm has the following input, output, and outer loop invariant.

\textbf{Input:} An arbitrary matroid $M_1$ and partition matroids $M_i$ for $2 \leq i \leq k$ on common ground set $X$.

\textbf{Output:} A $\left(1+\sum_{i=1}^k \left(\chi(M_i) -1\right)\right)$-coloring of $M = \cap_{i=1}^k M_i$.

\textbf{Outer Loop Invariant:} After $q$ iterations, $q$ elements of the ground set are feasibly colored. A coloring is \textit{feasible} if its color classes are independent sets in all of the matroids.

\medskip
\noindent \textbf{Description of our  Matroid Intersection Coloring Algorithm:} Iteratively update the current coloring $c$ along an arbitrary  source-sink path $P$, in the color-chordless subgraph $H$ of the current Edmonds digraph $G$ of $M_1$, with the property that $P$ is suffix-feasible with respect to each of the partition matroids $M_2, \ldots, M_k$. If all elements are colored, return the coloring.

To complete the description of our algorithm, we need to define $H$, what it means for a path $P$ to be suffix-feasible with respect to a matroid $M_i$,
and explain how to compute such a path $P$. It should be obvious from the definition of $H$ how to efficiently
compute $H$.
For ease of notation, let $\alpha = \chi(M_1)$ and $B = \sum_{i=2}^k \left(\chi(M_i)-1\right)$. 
Note that $\chi(M_1)$ can be computed using Edmonds' algorithm for single matroid coloring, and $\chi(M_i)$ for $2 \leq i \leq k$ is simply the ceiling of the maximum ratio of part size to part capacity over all parts in $M_i$.

\begin{definition}[Color-Chordless Subgraph]
The subgraph $H$ is computed from the Edmonds digraph $G=(V=[\alpha + B] \cup X,A)$ of $M_1$. The vertices in $H$ are partitioned
into layers $L_0, L_1, \ldots, L_h$, which are defined inductively as follows. 
The base layer is $L_0$, which is defined to be $[\alpha + B]$, the collection of color
nodes in $G$. There are no arcs entering $L_0$. Then, $L_i$ is defined
to be the collection of vertices $x$ not in  $L_0, \ldots, L_{i-1}$ such that $x$ has $B+1$ incoming arcs in $G$
from nodes of $B+1$ distinct colors in earlier layers of the subgraph $H$. In this definition, the color of an element node is its color in the coloring $c$. Let $y_1, \ldots, y_{B+1}$ be a set of such nodes (so $y_i$ and $y_j$ have distinct colors for $i \neq j$) with the property that for all $i \in [B+1]$, there does not exist a vertex $z_i$ with the same color as $y_i$, such that
 $z_i$ is in an earlier layer than $y_i$ and 
 $(z_i, x)$ is an arc in $G$. 
 Then the arcs $(y_1, x), \ldots, (y_{B+1}, x)$ 
 are added to $H$. Intuitively, the construction of 
 $H$ prefers arcs from earlier layers. Layer construction continues until a layer is empty, at which point it is not included and the process terminates.
\end{definition}

 \begin{definition}[Suffix-Feasible Path]
     \label{defn:suffixfeasible}
     Let $P = (x_1, x_2, \dots, x_{\ell} = u)$ be a path in the Edmonds digraph $G$, constructed from
     a feasible coloring $c$ of matroid $M_1=(X, \mathcal{I}_1)$, 
     that ends at a sink, that is an uncolored vertex $u$. Let $M_h$ be an arbitrary matroid.

     \begin{itemize}
         \item The updated coloring $c' = c \Delta P$ is defined as follows. In $c'$, 
a vertex $x_i$, $3 \le i \le \ell$, is  colored $c(x_{i-1})$. If $x_1$ is a color vertex, that is $x_1 \in [\alpha+B]$,
then $c'(x_2) = x_1$. If $x_1$ is not a color vertex, that is $x_1 \notin [\alpha+B]$, then $c'(x_2) = c(x_1)$ and $x_1$ is uncolored in $c'$, that
is $c'(x_1) = 0$.
\item
The path $P$ is suffix-feasible with respect to $M_h$ if 
     for every suffix $P_j = (  x_{\ell-j}, x_{\ell-j+1}, \dots, x_{\ell} = u)$ of $P$ it is the case 
     that $c \Delta P_j$ is a feasible coloring in $M_h$.
     \end{itemize}     
 \end{definition}

\paragraph{Algorithm to Compute $P$:}
The source-sink path  $P = (x_1=r, x_2, x_3, \dots, x_{\ell} = u)$ is computed iteratively.
The loop invariant is that after $j$ iterations,  the suffix $P_j$ will have been computed.
Initially $P_0$ is an arbitrary uncolored element $u \in U \cap V(H)$. 
We now explain how $P_{j}$ is computed from $P_{j-1}$. 
The vertex $x_{\ell-j}$ is set to an  arbitrary vertex $z$ in $H$ such that there is an arc $(z, x_{\ell-j+1})$ in $H$,
and the path $ (z, x_{\ell-j+1}, \ldots, x_\ell)$ is suffix-feasible with respect to each of the partition matroids $M_2, \ldots, M_k$. Once the vertex $x_{\ell-j}$ is set to a source vertex, this process terminates and $P := P_j$.

\subsection{Algorithm Analysis}
\label{subsec:intersectionanalysis}

In order to prove the correctness of our algorithm, we will prove the following statements. In Lemma \ref{lemma:feas_M1}, we will show all source-sink paths in $H$ are color-chordless paths in $G$. In Lemma \ref{lemma:main_struct}, we will show that $H$ contains an uncolored element $u$. In Lemma \ref{lemma:suffix-feasible}, we will show that during construction of the suffix-feasible path $P$, a vertex $z$ always exists to extend the suffix. In Lemma \ref{lemma:final}, we will tie everything together and show that our algorithm maintains a feasible coloring $c$ in $M_1, M_2, \dots, M_k$.

Throughout the analysis, we fix the following objects:

\begin{itemize}
    \item An arbitrary matroid $M_1$ and partition matroids $M_i$ for $2 \leq i \leq k$ on common ground set $X$.
    \item An arbitrary feasible coloring $c$ of $q < n$ elements of $M = \bigcap_{i=1}^k M_i$, using $(\alpha+B)$ colors where $\alpha = \chi(M_1)$ and $B = \sum_{i=2}^k (\chi(M_i)-1)$.
    \item The Edmonds digraph $G$ of the matroid $M_1$ with respect to the coloring $c$.
    \item The color-chordless subgraph $H$ in $G$ computed by our algorithm.
    \item The source-sink suffix-feasible path $P$ in $H$ computed by our algorithm. 
\end{itemize}

\begin{lemma}\label{lemma:feas_M1}
All source-sink paths in $H$ are color-chordless paths in $G$.
\end{lemma}

\begin{proof}

Let $P' = (x_1, x_2, \dots, x_\ell)$ be an arbitrary source-sink path in $H$. Assume for the sake of contradiction that there exists a color-chord of $P'$, that is, an arc $(x_j, x_k) \in A$ such that $1 \leq j \leq k-2$ and $x_j$ has the same color as $x_{k-1}$ in the current coloring $c$. Consider the time when $(x_{k-1}, x_k)$ was added to $H$. At this time, $x_{k-1}$ had minimum level among all $x \in V(H)$ with the same color as $x_{k-1}$ such that $(x, x_k) \in A$. But since $j \leq k-2$, the node $x_j$ has a lower level than $x_{k-1}$, a contradiction. Thus $P'$ is color-chordless. 

\end{proof}

We now state and prove Lemma \ref{lemma:main_struct}. The central claim is that if a cut in $G$ separates the color nodes from the uncolored elements, then it contains many arcs originating from different colors. Therefore $H$ must contain an uncolored element $u$, because otherwise it could be extended by one more level.

Recall that the set of nodes $L = [\alpha + B]$ are called \textbf{sources} and the set of uncolored elements $U$ are called \textbf{sinks}. We say that $(\overline{Y}, Y)$ is a \textbf{source-sink separating cut} if $L \subseteq V \setminus Y$ and $U \subseteq Y$. We denote by $\delta(\overline{Y},Y)$ the set of arcs going from $\overline{Y}$ to $Y$, that is $\delta(\overline{Y},Y) = \{(x, y) \in A : x \in \overline{Y}, y \in Y\}$.

\begin{lemma}\label{lemma:main_struct} Let $(\overline{Y}, Y)$ be a source-sink separating cut in $G$. Then there exists a subset of arcs $A' \subseteq \delta(\overline{Y}, Y)$ such that $|A'| > B \cdot |Y|$ and no two arcs in $A'$ share both the same target node $x$ and the same source node color $j$.
\end{lemma}

\begin{proof}
We will first show that many elements in $Y$ can change colors feasibly in $M_1$. We then observe that each one of these color changes induces a corresponding arc in $\delta(\overline{Y}, Y)$. 

By assumption, $M_1$ is an $\alpha$-colorable matroid. Let $T_1, T_2, \dots, T_\alpha$ be the color classes of an $\alpha$-coloring of $M_1$. Let $S_1, S_2, \dots, S_{\alpha+B}$ be the color classes of our $(\alpha+B)$-coloring $c$. Restrict both colorings to $Y$, by letting $T_i' = T_i \cap Y$ for $i = 1, 2, \dots, \alpha$ and $S_j' = S_j \cap Y$ for $j = 1, 2, \dots, \alpha+B$.

We now loop over $j = 1, 2, \ldots, \alpha+B$, and show that many elements $x \in Y \setminus S_j'$ have $S_j' \cup \{x\} \in \mathcal{I}_1$. For each $j = 1, 2, \dots, \alpha+B$, and for each $i = 1, 2, \dots, \alpha$, consider $S_j', T_i'$. There exist at least $|T_i'|-|S_j'|$ elements $x \in T_i' \setminus S_j' \subseteq Y \setminus S_j'$ such that $S_j' \cup \{x\} \in \mathcal{I}_1$ via the matroid exchange property (see Definition \ref{def:matroid}). Thus, for each $j = 1, 2, \dots, \alpha+B$, the number of elements $x \in Y \setminus S_j'$ such that $S_j' \cup \{x\} \in \mathcal{I}_1$ is at least $\sum_{i=1}^\alpha (|T_i'| - |S_j'|)$. Summing over all $j = 1, 2, \dots, \alpha+B$, the number of $j, x$ pairs with this property is at least

\begin{align*}
\sum_{j=1}^{\alpha+B} \sum_{i=1}^{\alpha} \left(|T_i'| - |S_j'|\right)
&= \sum_{j=1}^{\alpha+B} \sum_{i=1}^{\alpha} |T_i'| - \sum_{i=1}^{\alpha} \sum_{j=1}^{\alpha+B} |S_j'| \\
&= \sum_{j=1}^{\alpha+B} |Y| - \sum_{i=1}^{\alpha} |Y \setminus U| \\
&> \sum_{j=1}^{\alpha+B} |Y| - \sum_{i=1}^{\alpha} |Y| \\
&= (\alpha+B)|Y| - \alpha|Y| \\
&= B|Y|
\end{align*}

We now show that each $j, x$ pair such that $x \in Y \setminus S_j'$ and $S_j' \cup \{x\} \in \mathcal{I}_1$ induces a corresponding arc in $\delta(\overline{Y}, Y)$.

\begin{claim}\label{claim:main_struct_helper}
For all $j = 1, 2, \dots, \alpha+B$ and $x \in Y \setminus S_j'$, if $S_j' \cup \{x\} \in \mathcal{I}_1$, then there exists an arc $(y, x) \in \delta(\overline{Y}, Y)$ where $y$ has color $j$.
\end{claim}

\begin{proof}
If $S_j \cup \{x\} \in \mathcal{I}_1$, then $(j, x) \in \delta(\overline{Y}, Y)$. Else, $S_j \cup \{x\} \notin \mathcal{I}_1$. Thus, $S_j \cup \{x\}$ contains a unique circuit $C$ (Fact \ref{fact:unique_circuit}). This circuit must contain some element $y \in C \setminus Y$, for otherwise, we would have $C \subseteq Y$, implying $C \subseteq S_j' \cup \{x\}$, a contradiction. Thus $S_j - \{y\} \cup \{x\} \in \mathcal{I}_1$, so $(y, x) \in \delta(\overline{Y}, Y)$.
\end{proof}

Lastly, for each $j, x$ pair such that $x \in Y \setminus S_j'$ and $S_j' \cup \{x\} \in \mathcal{I}_1$, generate a unique arc $(y, x) \in \delta(\overline{Y}, Y)$ where $y$ has color $j$, and let $A' \subseteq \delta(\overline{Y}, Y)$ be the collection of these arcs. Then $|A'| > B|Y|$ and no two arcs in $A'$ share both the same target node $x$ and source node color $j$.
\end{proof}

We now remark how Lemma \ref{lemma:main_struct} shows that the color-chordless subgraph $H$ contains an uncolored element $u$. Assume for sake of contradiction that $H$ does not contain an uncolored element $u$. Let $Y = V \setminus V(H)$. Then $(\overline{Y}, Y)$ is a source-sink separating cut in $G$. By Lemma \ref{lemma:main_struct}, there must exist some node $x \in Y$ with at least $B+1$ incoming arcs from distinct colors in $V(H)$. But then $H$ could be extended by one more level, a contradiction.

We now show that during construction of the suffix-feasible path $P$, a vertex $z$ always exists to extend the suffix.

\begin{lemma}\label{lemma:suffix-feasible}
Let $P_j = (x_{\ell-j}, x_{\ell-j+1}, \dots, x_\ell=u)$ be the suffix-feasible path with respect to each of the partition matroids $M_2, M_3, \dots, M_k$ computed by our algorithm after $j$ iterations of the construction of $P$. Assume $x_{\ell-j}$ is not a source. Then there exists a vertex $z$ in $H$ such that there is an arc $(z, x_{\ell-j})$ in $H$ and $P_{j+1} = (z, x_{\ell-j}, x_{\ell-j+1}, \dots, x_\ell=u)$ is suffix-feasible with respect to each of the partition matroids $M_2, M_3, \dots, M_k$.
\end{lemma}

\begin{proof}
Let $Z$ be the set of vertices in $H$ with incoming arcs to $x_{\ell-j}$. Let $c' = c \Delta P_j$ (with the convention $c = c \Delta P_0 = c \Delta (u)$), and note $x_{\ell-j}$ is uncolored in $c'$. For each partition matroid $M_i$, let $C_i$ be $x_{\ell-j}$'s part in $M_i$. Each $C_i$ has at most $\chi(M_i) - 1$ color classes of $c'$ at the capacity constraint. Thus there are at most $B = \sum_{i=2}^k \left(\chi(M_i)-1\right)$ color classes of $c'$ which are at the capacity constraint in some $C_i$. Because $|Z| = B+1$, there exists some vertex $z \in Z$ with a color $r$ such that $r$'s color class in $c'$ is not at the capacity constraint in \textit{any} $C_i$. Let $P_{j+1} = (z, x_{\ell-j}, x_{\ell-j+1}, \dots, x_\ell=u)$. $c \Delta P_{j+1}$ is identical to $c \Delta P_j$, except that $x_{\ell-j}$ has color $r$ and $z$ is uncolored. By choice of $z$, $c \Delta P_{j+1}$ is a feasible coloring in partition matroids $M_2, M_3, \dots, M_k$. Further, all other suffixes of $P_{j+1}$ are suffixes of $P_j$. Thus $P_{j+1}$ is suffix-feasible with respect to partition matroids $M_2, M_3, \dots, M_k$.
\end{proof}

\begin{lemma}\label{lemma:final}
$c \Delta P$ is a feasible coloring of $q+1$ elements of $M = \bigcap_{i=1}^k M_i$ using $\alpha+B$ colors.
\end{lemma}

\begin{proof}
$P$ is a color-chordless path in $G$ by Lemma \ref{lemma:feas_M1}, so $c \Delta P$ is a feasible coloring of $q+1$ elements of $M_1$ by Lemma \ref{lemma:color-chordless_path}. $P$ is suffix-feasible with respect to partition matroids $M_2, M_3, \dots, M_k$ via Lemma \ref{lemma:main_struct} and Lemma \ref{lemma:suffix-feasible}, so $c \Delta P$ is a feasible coloring of $q+1$ elements of $M_2, M_3, \dots, M_k$.
\end{proof}

Starting with the empty coloring and applying Lemma~\ref{lemma:final} $n$ times proves Theorem~\ref{thm:main}.

\subsection{Runtime Analysis}\label{subsec:runtime}
We now prove the following theorem about the runtime of our matroid intersection coloring algorithm.

\begin{theorem}\label{thm:runtime}
The Matroid Intersection Coloring Algorithm described in Section \ref{sec:alg} can be implemented in time $O(n^3T(M_1))$, where $n = |X|$ and $T(M_1)$ is the runtime of a single independence oracle call to the general matroid $M_1$.
\end{theorem}

\begin{proof}
We will prove that on each iteration of the algorithm, the Edmonds digraph $G$ of $M_1$ can be computed in time $O(n^2T(M_1))$, that the color-chordless subgraph $H$ of $G$ can be computed in time $O(n^2)$, and that the suffix-feasible path $P$ in $H$ can be computed in time $O(n^2)$. Thus each iteration takes time $O(n^2 T(M_1))$, and the algorithm has runtime $O(n^3 T(M_1))$.

Without loss of generality, assume that $\alpha+B < n$ (otherwise we just output the singleton coloring of the elements) and that all partition matroids have chromatic number at least $2$ (otherwise we remove all partition matroids with chromatic number $1$, which are free matroids). The Edmonds digraph $G$ of $M_1$ can be constructed in time $O(n^2 T(M_1))$, because each pair of nodes requires an independence oracle call to the general matroid to determine whether an arc belongs between them. 

A naive construction of the color-chordless subgraph $H$ takes $O(n^3)$ time, but we can construct it in time $O(n^2)$ in the following way. The idea is to maintain a 2D array $D$ which stores, for each $x$ not already added to a layer, earliest-layer arcs $(y, x)$ from distinct colors. This allows us to perform only a single pass over the arcs during the entire construction of $H$. Specifically, we maintain an $n \times (\alpha+B)$ array $D$, a length $n$ array $S$, and a subset $W \subseteq X$ with the following properties before the construction of layer $L_i, i \geq 0$.

\begin{itemize}
    \item For each $x \notin L_0 \cup L_1 \cup \dots \cup L_{i-1}$ and $j \in [\alpha+B]$, $D[x][j] = y$ if $(y, x) \in A$, $y = j$ or $c(y) = j$, $y \in L_0 \cup L_1 \cup \dots \cup L_{i-1}$, and $y$ is from an earliest layer with these properties. If no such $y$ exists, then $D[x][j] = 0$.
    \item For each $x \notin L_0 \cup L_1 \cup \dots \cup L_{i-1}$, $S[x]$ is the number of nonzero entries $D[x][j]$ over $j \in [\alpha+B]$.
    \item $W = \{x \notin L_0 \cup L_1 \cup \dots \cup L_{i-1} : S[x] \geq B+1\}$
\end{itemize}

Initialize $D = 0, S = 0, W = \emptyset$ before the construction of layer $L_0$. To construct layer $L_i, i \geq 1$, set $L_i \leftarrow W$ (set $L_0 \leftarrow [\alpha+B]$). Then for each $x \in W$, iterate over $D[x][j]$ for $j \in [\alpha+B]$, and add $B+1$ distinct arcs $(y, x)$ where $D[x][j] = y$. To update $D, S, W$, set $W \leftarrow \emptyset$. Iterate over $x \in L_i$, and suppose $j$ is the color of $x$. For each arc $(x, z) \in A$, if $z \notin L_0 \cup L_1 \cup \dots \cup L_i$ and $D[z][j] = 0$, update $D[z][j] = x$ and increment $S[z]$. If $S[z] \geq B+1$, add $z$ to $W$.

It is clear from the algorithm description that if $D, S, W$ have their invariants, the construction of $L_i, i \geq 0$ is done correctly. Further, $D, S, W$ maintain their invariants. Thus the above algorithm correctly constructs the color-chordless subgraph $H$.

To bound the runtime, observe that construction of layer $L_i, i \geq 0$ takes $O(|L_i| \cdot (\alpha+B)) = O(|L_i|n)$ time, and updating of $D, S, W$ involves a single pass over the outgoing arcs from $L_i$. Thus the total runtime is $O\left(\sum_i|L_i|n\right) + O(|A|) = O(n^2)$.

The suffix-feasible path $P$ in $H$ can be constructed in time $O(n^2)$. Note that $P$ has length at most $n$. To extend the current path $P_j = (x_{\ell-j}, x_{\ell-j+1}, \dots, x_{\ell} = u)$ to $P_{j+1}$, we follow the proof of Lemma \ref{lemma:suffix-feasible}.

Compute the set $R \subseteq [\alpha+B]$ of color classes which are at the capacity constraint in $c' = c \Delta P_j$ in $x_{\ell-j}$'s part in some partition matroid. To compute $R$ in $O(n)$ time, our algorithm also maintains, during its entire execution, the set of ``saturated'' colors in each part in each partition matroid (i.e. those colors which are at the capacity constraint), and the number of elements of each color $i \in [\alpha+B]$ in each part in each partition matroid. This requires $O(kn(\alpha+B)) = O(n^3)$ time to initialize (which is fine under the theorem's stated runtime guarantee) and $O(n)$ time to update on a single recoloring. Given this data, the algorithm can compute $R$ by iterating over the ``saturated'' sets for $x_{\ell-j}$'s parts in the partition matroids, and the total number of saturated color classes is at most $\sum_{i=2}^k \left(\chi(M_i)-1\right) = B = O(n)$.

Iterate over the incoming arcs to $x_{\ell-j}$ in $H$ and find an element $z$ s.t. $(z, x_{\ell-j}) \in H$ and $z$ does not have a color in $R$. To extend $P_j$ to $P_{j+1}$ requires $O(n)$ time to construct $R$ and $O(n)$ time to iterate over the arcs incoming to $x_{\ell-j}$. Thus it takes $O(n)$ time to extend $P_j$ to $P_{j+1}$, so it takes $O(n^2)$ time to construct $P$.

Thus on each iteration, it takes $O(n^2T(M_1))$ time to construct $G$, $O(n^2)$ time to construct $H$, and $O(n^2)$ time to construct $P$, so it takes $O(n^2T(M_1))$ time in total. Thus the algorithm runs in time $O(n^3T(M_1))$.
\end{proof}

\subsection{Coloring a Matroid Intersection on an Example}
\label{app:intersectionexample}

Consider the following matroids $M_1, M_2$ and current $3$-coloring of elements $c$ in Figure \ref{fig:sec3_init_coloring}. $M_1$ is the laminar matroid from the previous example and $M_2$ is a partition matroid.   Color $1$ is blue (vertical hatches), color $2$ is red (diagonal hatches), and color $3$ is green (horizontal hatches). $c(x_6) = 0, c(x_5) = 1, c(x_3) = c(x_4) = 2$, and $c(x_1) = c(x_2) = 3$. $S_1 = \{x_5\}, S_2 = \{x_3, x_4\}, S_3 = \{x_1, x_2\}$, and $U = \{x_6\}$. Note $S_1, S_2, S_3 \in \mathcal{I}_1 \cap \mathcal{I}_2$. $\chi(M_1) = \chi(M_2) = 2$, so $\alpha = 2$ and $B = 1$.

\begin{figure}[h]
\begin{center}
\begin{subfigure}[t]{0.3\textwidth}
\includegraphics[width=\textwidth]{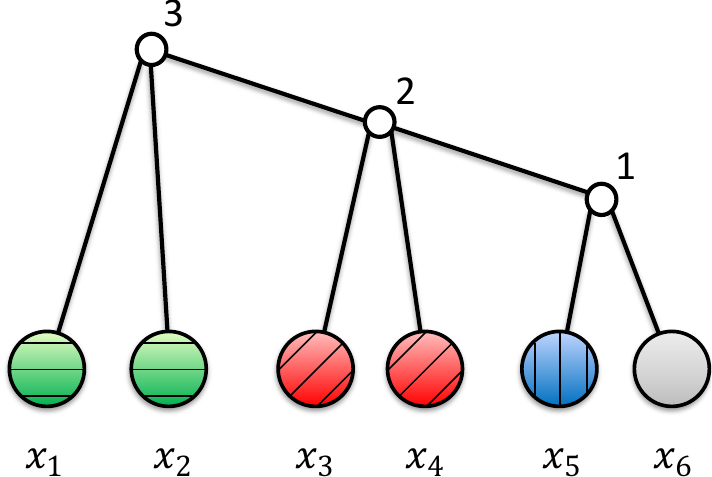}
\caption{}
\label{fig:sec3_laminar}
\end{subfigure}
\qquad
\begin{subfigure}[t]{.3\textwidth}
\includegraphics[width=\textwidth]{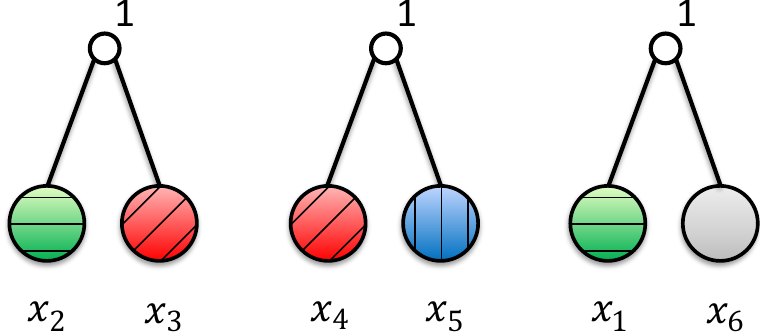}
\caption{}
\label{fig:sec3_partition}
\end{subfigure}
\end{center}
\caption{(\ref{fig:sec3_laminar}) is the laminar matroid $M_1$ and (\ref{fig:sec3_partition}) is the partition matroid $M_2$.}
\label{fig:sec3_init_coloring}
\end{figure}

The construction of the Edmonds digraph $G$ of $M_1$ is shown in Figure \ref{fig:sec3_G1G2G3G}.  Note $S_1 \cup \{x_1\} = \{x_1, x_5\} \in \mathcal{I}_1$, so there is an arc $(1, x_1) \in A_1$. Note $S_2 \cup \{x_6\} = \{x_3, x_4, x_6\} \notin \mathcal{I}_1$ but $S_2 - \{x_3\} \cup \{x_6\} = \{x_4, x_6\} \in \mathcal{I}_1$, so there is an arc $(x_3, x_6) \in A_2$. Note $S_3 \cup \{x_4\} = \{x_1, x_2, x_4\} \in \mathcal{I}_1$, so there is an arc $(3, x_4) \in A_3$.

\begin{figure}[h]
\begin{center}
\begin{subfigure}[t]{0.2\textwidth}
\includegraphics[width=\textwidth]{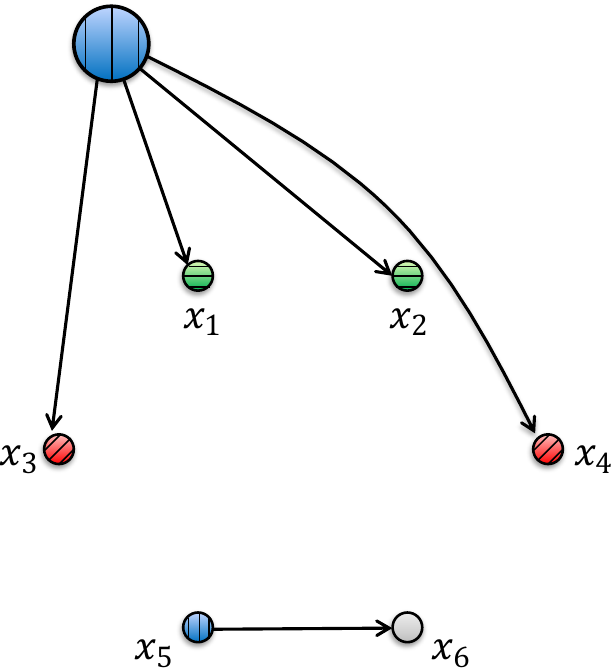}
\caption{}
\label{fig:sec3_G2.1}
\end{subfigure}
\hspace{.5cm}
\begin{subfigure}[t]{0.2\textwidth}
\includegraphics[width=\textwidth]{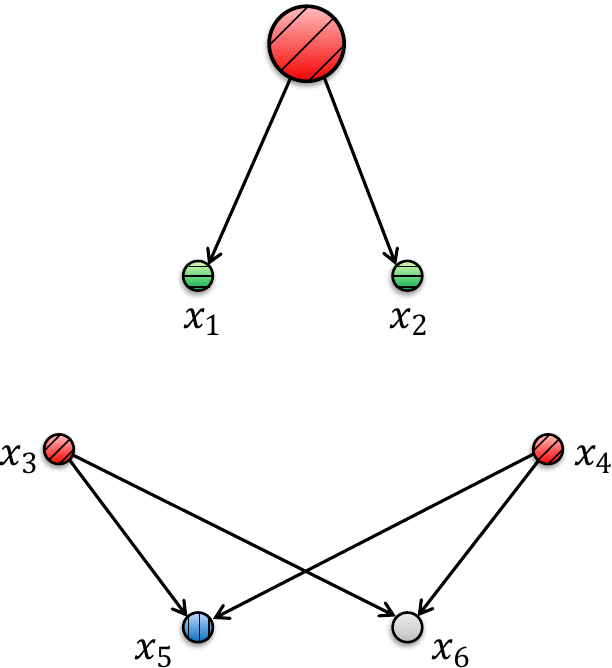}
\caption{}
\label{fig:sec3_G2.2}
\end{subfigure}
\hspace{0.5cm}
\begin{subfigure}[t]{0.2\textwidth}
\includegraphics[width=\textwidth]{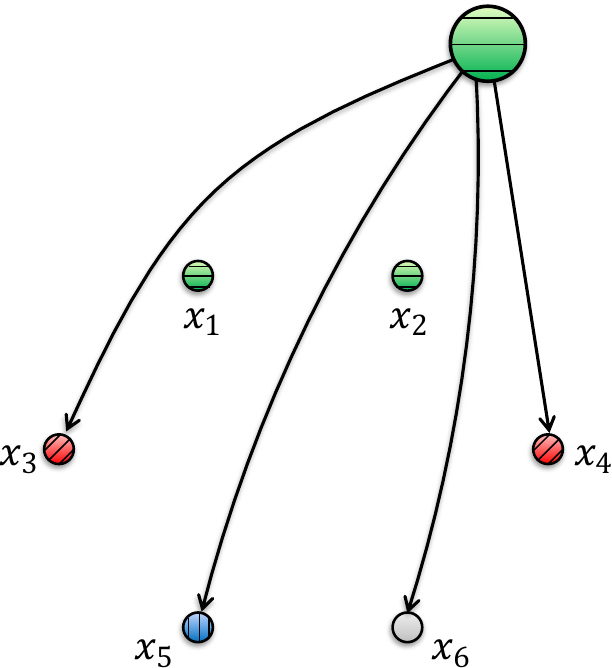}
\caption{}
\label{fig:sec3_G2.3}
\end{subfigure}
\hspace{0.5cm}
\begin{subfigure}[t]{0.2\textwidth}
\includegraphics[width=\textwidth]{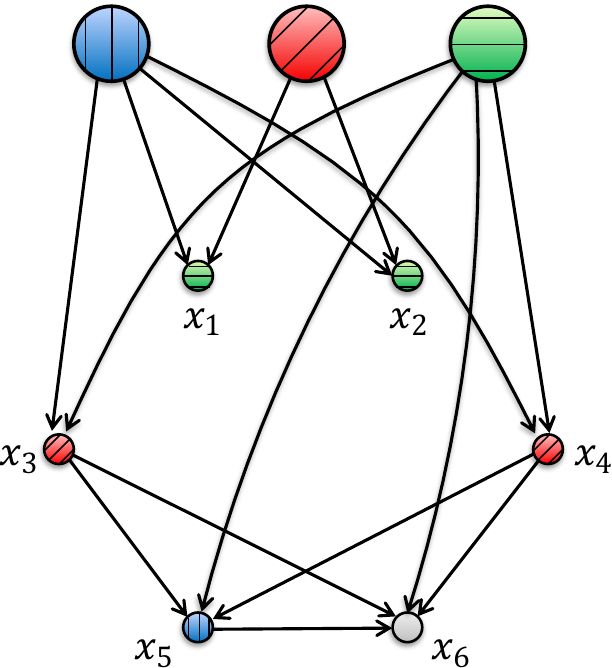}
\caption{}
\label{fig:sec3_G2.0}
\end{subfigure}
\end{center}
\caption{(\ref{fig:sec3_G2.1}) is edges $A_1$, (\ref{fig:sec3_G2.2}) is edges $A_2$, (\ref{fig:sec3_G2.3}) is edges $A_3$, and (\ref{fig:sec3_G2.0}) is the Edmonds digraph $G$.}
\label{fig:sec3_G1G2G3G}
\end{figure}

The construction of the color-chordless subgraph $H$ is shown in Figure \ref{fig:sec3_H}. Observe that the arcs are directed down the levels, the nodes not on $L_0$ have in-degree $2$, and the incoming arcs to nodes not on $L_0$ have different source colors. Observe that $x_5$ and $x_6$ are not included on $L_1$ because, although they have an incoming arc from the Green $= 3$ node, they do not have two incoming arcs from $L_0$. Note that on $L_2$, the incoming arcs to $x_5, x_6$ could be from either $x_3$ or $x_4$. 

The construction of the suffix-feasible path $P$ with respect to $M_2$ is shown in Figure \ref{fig:sec3_P1} and Figure \ref{fig:sec3_P2}. Initially, $P_0 = x_6$. Figure \ref{fig:sec3_P1} shows the suffix $P_1 = (x_3, x_6)$ and the feasible coloring $c\Delta P_1$ in $M_2$. Note that $P_1' = (3, x_6)$ is \textit{not} suffix-feasible with respect to $M_2$, because $c \Delta P_1'$ is not a feasible coloring in $M_2$ (both $x_1, x_6$ are Green in $c \Delta P_1'$). Figure \ref{fig:sec3_P2} shows the suffix $P_2 = (1, x_3, x_6)$, which is the complete suffix-feasible path $P$, and the feasible coloring $c\Delta P_2$ in $M_2$. Note that $P_2' = (3, x_3, x_6)$ is \textit{not} suffix-feasible with respect to $M_2$, because $c \Delta P_2'$ is not a feasible coloring in $M_2$ (both $x_2, x_3$ are Green in $c \Delta P_2'$).

\begin{figure}[h]
\begin{center}
\begin{subfigure}[t]{.25\textwidth}
\includegraphics[width=\textwidth]{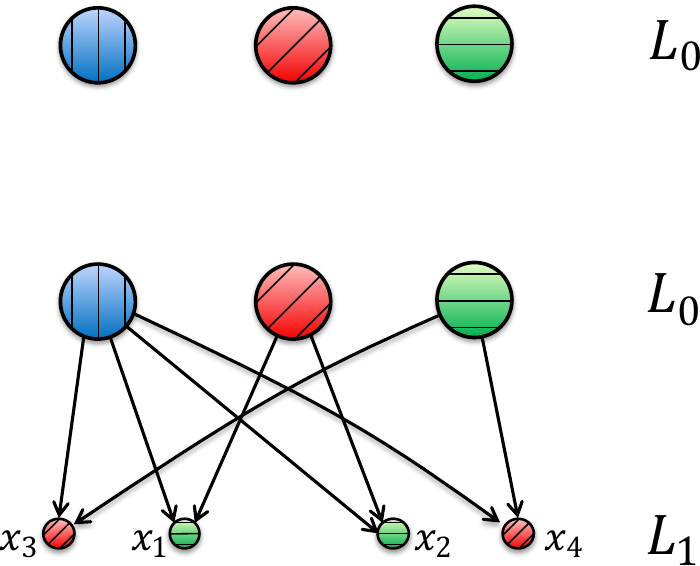}
\end{subfigure}
\begin{subfigure}[t]{.25\textwidth}
\includegraphics[width=\textwidth]{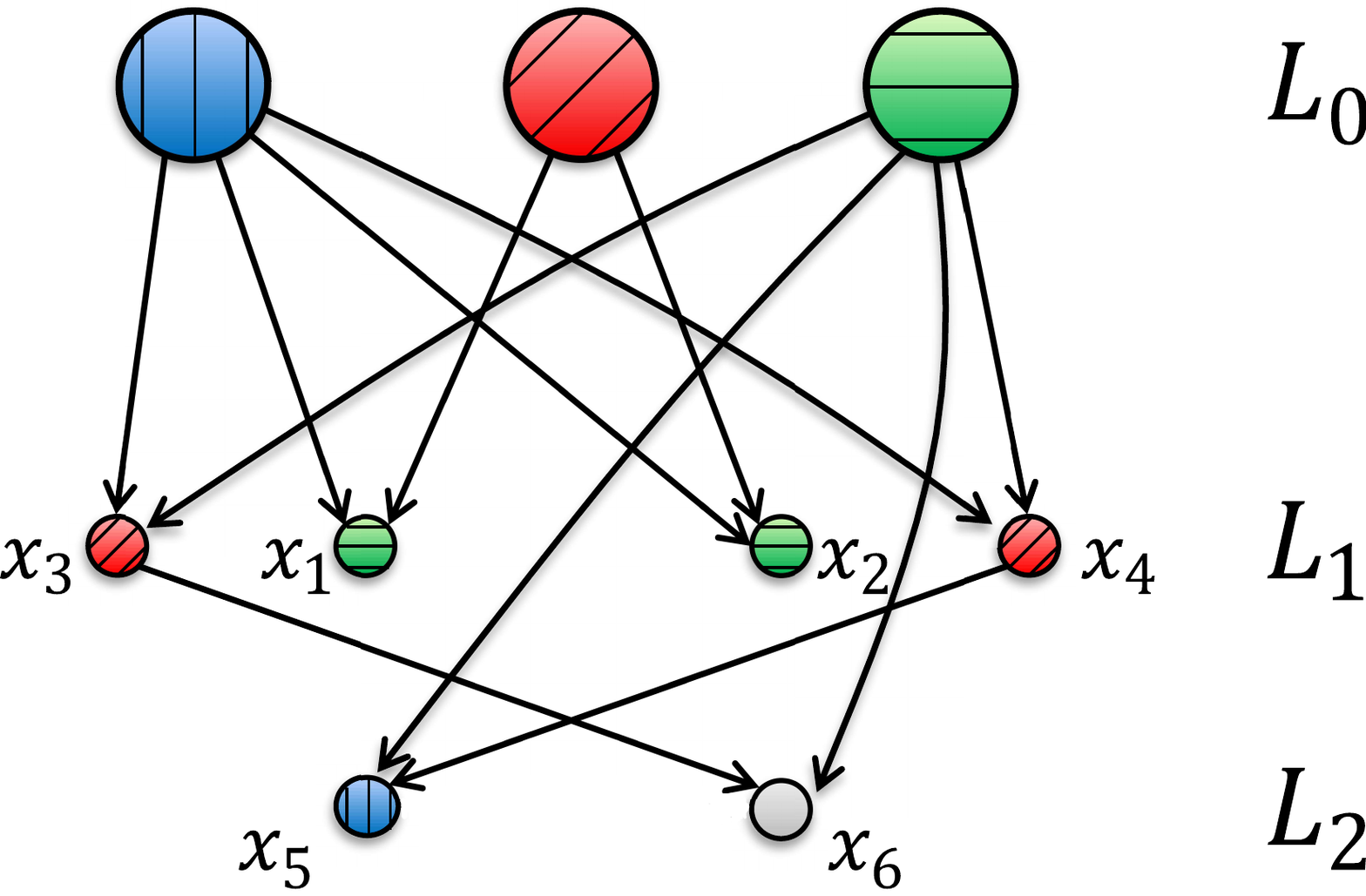}
\end{subfigure}
\end{center}
\caption{Constructing the color-chordless subgraph $H$ of $G$.}
\label{fig:sec3_H}
\end{figure}

\begin{figure}[h]
\centering
\begin{subfigure}[t]{.3\textwidth}
\includegraphics[width=\textwidth]{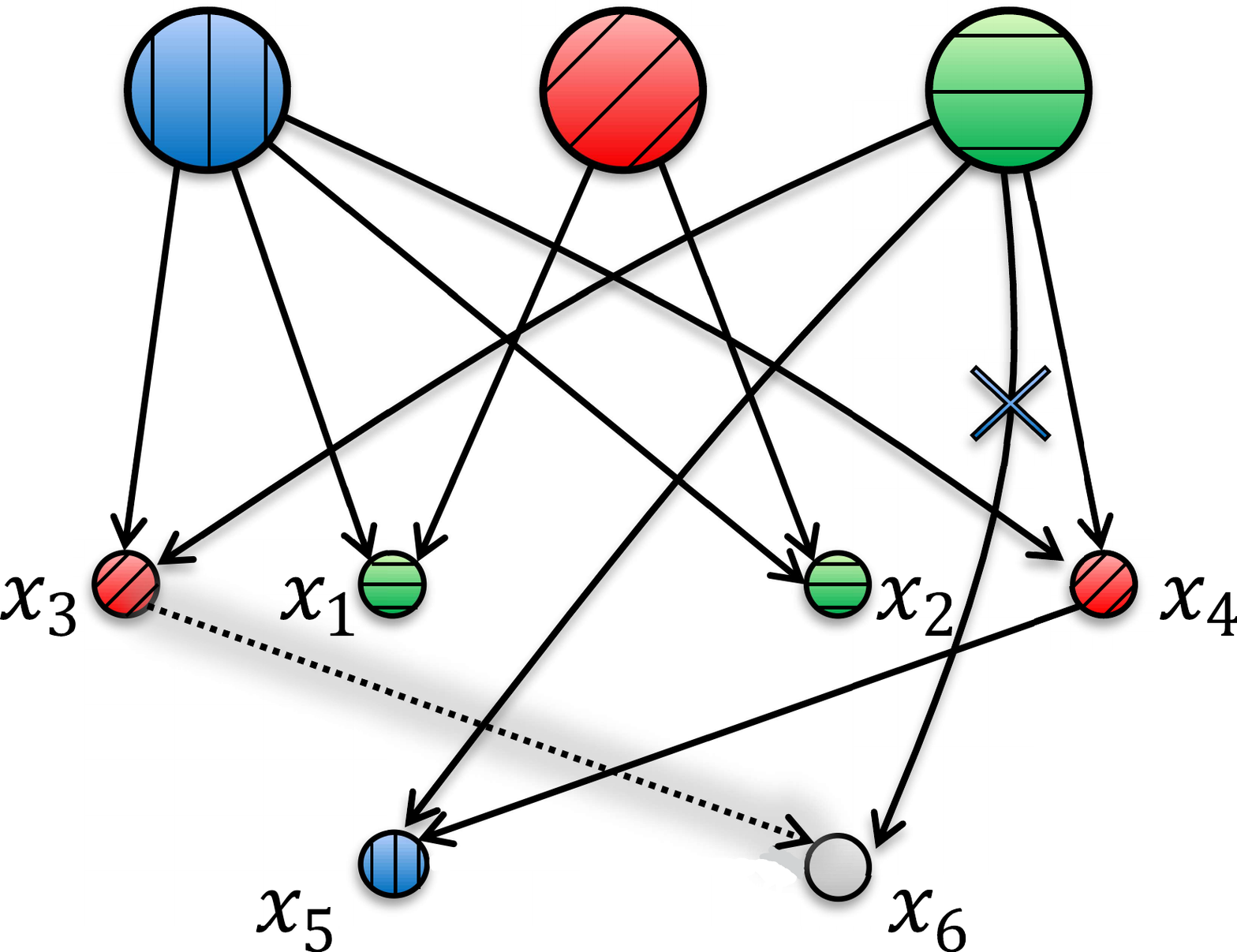}
\caption{}
\label{fig:sec3_path1_H}
\end{subfigure}
\qquad
\begin{subfigure}[t]{.3\textwidth}
\includegraphics[width=\textwidth]{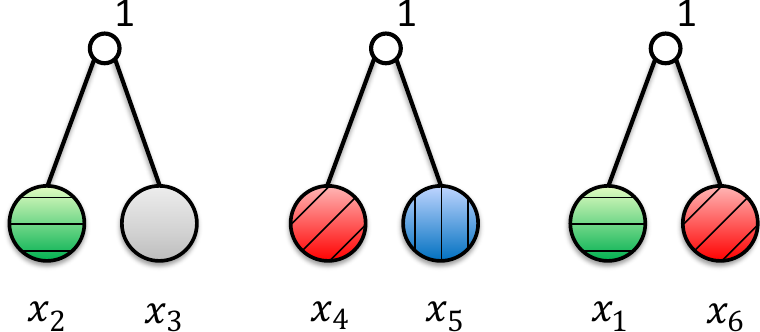}
\caption{}
\label{fig:sec3_path1_partition}
\end{subfigure}
\caption{(\ref{fig:sec3_path1_H}) is the suffix $P_1$. (\ref{fig:sec3_path1_partition}) is $c\Delta P_1$ on $M_2$.}
\label{fig:sec3_P1}
\end{figure}

\begin{figure}
\centering
\begin{subfigure}[t]{.3\textwidth}
\includegraphics[width=\textwidth]{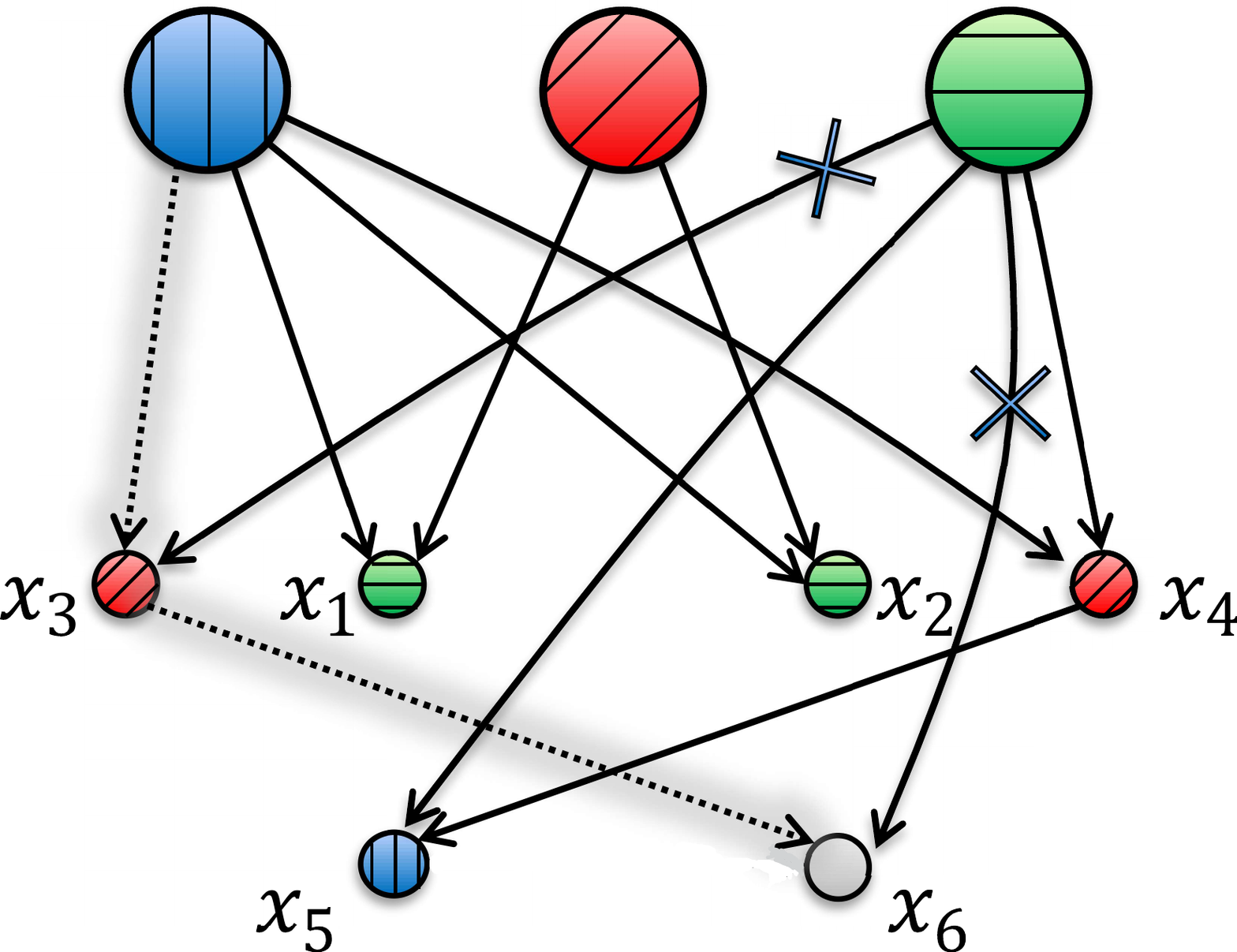}
\caption{}
\label{fig:sec3_path2_H}
\end{subfigure}
\qquad
\begin{subfigure}[t]{.3\textwidth}
\includegraphics[width=\textwidth]{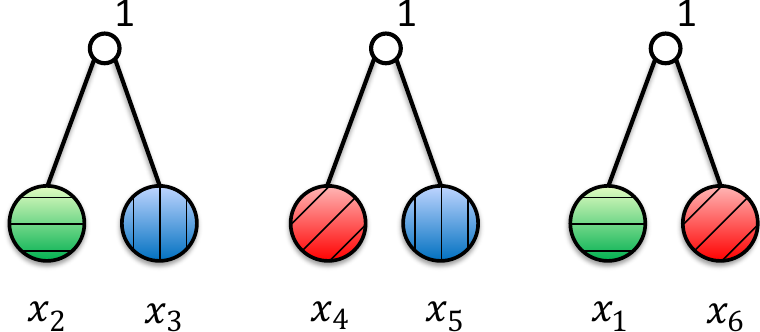}
\caption{}
\label{fig:sec3_path2_partition}
\end{subfigure}
\caption{(\ref{fig:sec3_path2_H}) is the suffix $P_2$ (the final path $P$). (\ref{fig:sec3_path2_partition}) is $c\Delta P_2$ on $M_2$.}
\label{fig:sec3_P2}
\end{figure}

\section{Hardness}\label{app:hardness}
In this section, we show that it is NP-hard to approximate the coloring of $k$ \emph{partition} matroids to a $k^{1-\epsilon}$ factor for any $\epsilon >0$ when $k$ is polynomial in the number of elements. We will show an approximation-preserving reduction from vertex coloring a graph with maximum degree $k-1$ to coloring the intersection of $k$ partition matroids.

\begin{theorem}\label{thm:hardness}
Consider the following problems:
\begin{enumerate}
\item[(A)] Coloring the intersection of $k$ partition matroids on a common ground set, where each part has size at most $2$ and capacity $1$.
\item[(B)] Vertex coloring graphs with maximum degree at most $k-1$.
\end{enumerate}
There is a polynomial-time approximation-preserving reduction showing a $c$-approximation algorithm for problem (A) yields a $c$-approximation algorithm for problem (B).
\end{theorem}

\begin{proof}

\noindent\textbf{Reducing (B) to (A).}
Let $G = (V,E)$ be a graph with maximum degree at most $k-1$. By Vizing's Theorem, the edges of $G$ can be decomposed into $k$ disjoint matchings $S_1, S_2, \dots, S_k$ in polynomial time~\cite{Vizing1964,MisraGries1992}.

For each $i \in [k]$, we construct a partition matroid $M_i = (V, \mathcal{I}_i)$ as follows:
\begin{itemize}
\item For each edge $\{u,v\} \in S_i$, create a part $\{u,v\}$ with capacity $1$.
\item For each vertex $v \in V$ not covered by any edge in $S_i$, create a singleton part $\{v\}$ with capacity $1$.
\end{itemize}

Since the parts partition $V$ and each has capacity $1$, this defines a valid partition matroid. Moreover, each part has size at most $2$.

\begin{claim} A set $I \subseteq V$ is independent in $\bigcap_{i=1}^k M_i$ if and only if $I$ is an independent set of vertices in $G$.
\end{claim}

\begin{proof}
($\Rightarrow$) Suppose $I \in \bigcap_{i=1}^k \mathcal{I}_i$. Let $\{u,v\} \in E$ be an arbitrary edge. Since $S_1, \ldots, S_k$ partition $E$, there exists some $j \in [k]$ such that $\{u,v\} \in S_j$. Thus $\{u,v\}$ forms a part in matroid $M_j$. Since $I$ is independent in $M_j$ and this part has capacity $1$, we have $|I \cap \{u,v\}| \leq 1$. Therefore, $u$ and $v$ cannot both be in $I$. Since $\{u,v\}$ was arbitrary, $I$ contains no edge of $G$, so $I$ is an independent set of vertices in $G$.

($\Leftarrow$) Suppose $I$ is an independent set of vertices in $G$. Fix any $i \in [k]$ and consider an arbitrary part $P$ in $M_i$. If $P$ is a singleton, then trivially $|I \cap P| \leq 1$. If $P = \{u,v\}$ for some edge $\{u,v\} \in S_i \subseteq E$, then since $I$ is independent in $G$, at most one of $u,v$ is in $I$, so $|I \cap P| \leq 1$. Since this holds for all parts in all matroids, $I \in \bigcap_{i=1}^k \mathcal{I}_i$.
\end{proof}

Since the independent sets of $G$ and $\bigcap_{i=1}^k M_i$ coincide, their chromatic numbers are equal:
$$\chi(G) = \chi\left(\bigcap_{i=1}^k M_i\right).$$

\medskip
\noindent\textbf{Approximation preservation.}
The reduction can be computed in polynomial time. Moreover, given a feasible coloring for the matroid problem, the same coloring is a feasible coloring for the corresponding instance of the graph problem (since the independent sets coincide). 
\end{proof}

There is a trivial polynomial-time approximation-preserving reduction from coloring graphs on $k$ vertices to coloring graphs on $n = \text{poly}(k)$ vertices with maximum degree $k-1$: simply add $n-k$ isolated vertices to the original graph.

\cite{Zuckerman2006LinearDegreeExtractors} shows that it is NP-hard to approximate the chromatic number of a graph on $k$ vertices within a factor of $k^{1-\epsilon}$. Thus it is NP-hard to approximate the chromatic number of a graph on $n$ vertices with maximum degree $k-1 = \text{poly}(n)$ within a factor of $k^{1-\epsilon}$. Finally, we obtain the following corollary for the hardness of matroid intersection coloring.

\begin{corollary}\label{cor:hardness}
For any $\epsilon > 0$, it is NP-hard to approximate the chromatic number of the intersection of $k$ partition matroids within a factor of $k^{1-\epsilon}$ when $k$ is polynomial in the number of elements.
\end{corollary}

Corollary \ref{cor:hardness} shows that a polynomial-time $O(k)$-approximation for matroid intersection coloring is essentially best possible when the number of matroids is large compared to the number of elements. It is possible that one could obtain an $o(k)$-approximation when $k = \text{polylog}(n)$, for example, but one is still bounded by more complicated vertex coloring hardness results (see \cite{Khot2001ImprovedInapproximability}).

\section{Applications and Extensions}\label{app:applications}

This section demonstrates the application of Theorem~\ref{thm:main} to derive polynomial-time algorithms for two combinatorial problems and an extension of Theorem~\ref{thm:main} to list coloring. In Section \ref{subsec:gen_rota}, we consider a generalization of Rota's Basis Conjecture. In Section \ref{subsec:graphmatroidintersection}, we consider coloring the intersection of the independent sets of a graph and a matroid. In Section \ref{app:lc}, we extend Theorem~\ref{thm:main} to list coloring.

\subsection{Generalized Rota's Basis Conjecture}\label{subsec:gen_rota}

In this subsection, we show how Theorem \ref{thm:main} makes progress on a generalization of Rota's Basis Conjecture. We emphasize that Theorem \ref{thm:main} does \emph{not} improve the state-of-the-art for the actual Rota's Basis Conjecture.

\begin{conjecture}[Generalized Rota's Basis Conjecture]\label{conj:gen_rota}
 If  
the ground set $X$ of a rank $r$ matroid $M = (X, \mathcal{I})$ can be partitioned into $m$ monochromatic
bases $B_1, \ldots, B_m$, then $X$ can be partitioned into $\max\{m, r\}$ rainbow
independent sets $R_1, \ldots, R_{\max\{m, r\}}$.
\end{conjecture}

In the context of Conjecture \ref{conj:gen_rota}, a subset of elements is ``rainbow'' if it contains no two elements of the same color, where the coloring on the elements is given by the monochromatic bases $B_1, \ldots, B_m$. We prove the following result which makes constructive progress on Conjecture \ref{conj:gen_rota} by reducing the setting to coloring the intersection of a general matroid and a partition matroid.

\begin{theorem}\label{thm:gen_rota}
 If  
the ground set $X$ of a rank $r$ matroid $M = (X, \mathcal{I})$ can be partitioned into $m$ monochromatic
bases $B_1, \ldots, B_m$, then $X$ can be partitioned into $m+r-1$ rainbow
independent sets $R_1, \ldots, R_{m+r-1}$.
\end{theorem}

\begin{proof}
Let $M_1 = M$ and define partition matroid $M_2 = (X, \mathcal{I}_2)$ where the parts of $M_2$ are $B_1, B_2, \dots, B_m$ and the capacity of each part is $1$. Apply Theorem \ref{thm:main} to produce a coloring of $M_1 \cap M_2$. We have $\chi(M_1) = m$ and $\chi(M_2) = r$, so this coloring uses $1+\left(\chi(M_1)-1\right) + \left(\chi(M_2) - 1\right) = m+r-1$ colors. Further, the independent sets in $M_1 \cap M_2$ are rainbow independent sets of $M$ in the original setting. Thus the coloring of $M_1 \cap M_2$ gives a partition of $X$ into $m+r-1$ rainbow independent sets.
\end{proof}

An earlier bound follows from Theorem \ref{thm:nonconstructive2}. Applying the same reduction from Conjecture \ref{conj:gen_rota} to coloring the intersection of a general matroid and a partition matroid, Theorem \ref{thm:nonconstructive2} shows that $X$ can be (nonconstructively) partitioned into $m+r$ rainbow independent sets $R_1, \dots, R_{m+r}$. Thus our main contribution is a constructive version of this result (and a modest improvement).

Rota's Basis Conjecture is Conjecture \ref{conj:gen_rota} in the special case $m=r$.

\begin{conjecture}[Rota's Basis Conjecture] \cite{rota_basis}
    \label{conj:rota}
 If  
the ground set $X$ of a rank $r$ matroid $M = (X, \mathcal{I})$ can be partitioned into $r$ monochromatic
bases $B_1, \ldots, B_r$, then $X$ can be partitioned into $r$ rainbow
bases
$R_1, \ldots, R_r$.
\end{conjecture}

The combinatorial literature on Conjecture \ref{conj:rota} is significant. We just touch on a few key results here. Theorem \ref{thm:nonconstructive2} proves that $X$ can be partitioned into $2r$ rainbow independent sets. \cite{halfway_rota} proved that the matroid admits $(1/2-o(1))r$ disjoint rainbow bases. \cite{Pokrovskiy2025} proved that the matroid admits $r-o(r)$ disjoint rainbow independent sets of size $r-o(r)$. \cite{MontgomerySauermann2025RotaPackingCovering} proved that the matroid admits $r-o(r)$ disjoint rainbow bases, and can be covered with $r+o(r)$ rainbow independent sets. We believe the aforementioned results are nonconstructive -- at minimum, it is not explicitly stated whether the results yield polynomial-time algorithms.

The 12th Polymath Project was dedicated to Rota's Basis Conjecture, where it was shown that $X$ can be partitioned into $2r-2$ rainbow independent sets constructively for $r \geq 2$ \cite{rota_polymath}. After submission of the conference version of this work, it was shown that $X$ can be partitioned into $r+o(r)$ rainbow independent sets constructively, matching the best-known nonconstructive bound in the covering setting \cite{arndt2026approximation}. In fact, this constructive asymptotic result extends to the setting of Conjecture \ref{conj:gen_rota} under the condition $m = \omega(\log r)$. 

The constructive $2r-2$ bound \cite{rota_polymath} is strictly better than our result applied to the Rota setting $m=r$, which shows that $X$ can be partitioned into $2r-1$ independent sets constructively. Further, the subsequent asymptotically optimal covering result \cite{arndt2026approximation} is better than our result from Theorem \ref{thm:gen_rota} for some parameter regimes of $m, r$. Thus our main additional contribution is an \emph{unconditional} constructive statement about Conjecture \ref{conj:gen_rota}, but not Conjecture \ref{conj:rota}.

\subsection{Graph-Matroid Intersection Coloring}\label{subsec:graph-matroid}
\label{subsec:graphmatroidintersection}

In this subsection, we show how Theorem \ref{thm:main} makes progress on a generalization of the Strong Coloring Conjecture. An instance of graph-matroid intersection coloring consists of a simple graph $G=(V, E)$, and a matroid 
$M=(V, \mathcal{I})$, whose ground set is the vertex set of the graph.
Let $\mathcal J$ be the collection of graph \emph{independent sets} in $G$, i.e. a subset $J \subseteq V$ is independent if no two vertices in $J$ are adjacent in $G$. We can  equivalently view $G$ as the set system $(V, \mathcal{J})$. 
We consider problems related to coloring the intersection of $G$ and $M$, where
$G \cap M := (V, \mathcal{J} \cap \mathcal{I})$, using as few colors as possible.

Haxell~\cite{Haxell2011} introduced the \emph{Happy Dean problem} to illustrate
coloring the intersection of a graph with a partition matroid.
Let each vertex represent a faculty member, and let an edge \((v,w)\) in
a graph \(G\) signify that professors \(v\) and \(w\) cannot attend the same meeting.
The accompanying partition matroid \(M\) encodes departmental structure: the parts are the departments, and a set is independent in \(M\) if it
contains at most one member from each department.
The dean wishes to schedule the minimum number of meeting slots so that: (1) every faculty member is assigned to exactly one slot, (2) no two conflicting faculty meet together (\(G\)-independence), and (3)  no two faculty from the same department meet together (\(M\)-independence). Assigning a distinct \emph{color} to each meeting slot, this scheduling task is
equivalent to coloring the intersection \(G \cap M\):
each color class must be independent in \(G\) (no edge)
and in \(M\) (at most one vertex per partition part).

The following is the well-known Strong Coloring Conjecture~\cite{scc}. 

\begin{conjecture} [Strong Coloring Conjecture]
\label{conj:strongcoloringconjecture}
If $M$ is a partition matroid with capacities of $1$ and chromatic number at most $2 \Delta$, and $G$ is
   a graph with maximum degree at most $\Delta$,
   then $\chi(M \cap G) \le 2 \Delta$.
\end{conjecture}

Haxell \cite{haxell_scn} gave a nonconstructive proof of  the following weakening of the Strong Coloring Conjecture:

\begin{theorem}
\label{thm:haxellstrongcoloring}
If $M$ is a partition matroid with capacities of $1$, and $G$ is a graph with maximum degree at most $\Delta$, then $\chi(M \cap G) \le \max(3 \Delta(G) -1, \chi(M))$.
\end{theorem}

A corollary of Theorem \ref{thm:haxellstrongcoloring}, that more closely matches the Strong Coloring Conjecture is:
if $M$ is a partition matroid with chromatic number at most $3 \Delta -1$ and $G$ is
   a graph with maximum degree at most $\Delta$
   then $\chi(M \cap G) \le 3 \Delta-1$. There is also a constructive version of Haxell's result in the
   literature. 

   \begin{theorem}~\cite{GrafH20,GrafHH22} 
   \label{thm:effectivecoloring}
There is a randomized algorithm with expected polynomial runtime that, given a partition matroid $M$ with capacities of $1$ and a graph
$G$ with maximum degree at most $\Delta$, produces a coloring of $G \cap M$ using
at most $\max\left\{\chi(M), \lceil(3+\varepsilon)\Delta\rceil\right\}$ colors for any fixed constant $\varepsilon > 0$.
 \end{theorem}

The algorithm in \cite{GrafH20,GrafHH22} combines an
augmenting-tree local-search routine with
an iterative random sparsification routine  guided by the constructive Lovász Local Lemma. 
It was conjectured in \cite{ab06} that the Strong Coloring Conjecture still holds if the
condition that $M$ is a partition matroid is relaxed to allow $M$ to be an arbitrary matroid.  

\begin{conjecture}\label{conj:Aharoni}
If $M$ is an arbitrary  matroid with chromatic number at most $2 \Delta$ and $G$ is a graph with maximum degree at most $\Delta$    then $\chi(M \cap G) \le 2 \Delta$.
\end{conjecture}

There is a simple counterexample to this conjecture when $\Delta = 1$, formed by taking $M$ to be the graphic matroid on the six edges of $K_4$, and $G$ to be a matching on six vertices where the three opposite pairs of edges in $K_4$ (the vertices of $G$) are matched. The conjecture is that $\chi(M \cap G) \leq 2$, but it is easy to see that $\chi(M \cap G) = 3$. However, one can still ask what can be salvaged from Conjecture \ref{conj:Aharoni}, and we can apply Theorem~\ref{thm:main} to graph-matroid intersection coloring to provide an answer.

\begin{theorem}
\label{thm:ourstrongcoloring}
There is a polynomial-time algorithm that given a general matroid $M$ and a graph
$G$ with maximum degree at most $\Delta$ produces a coloring of $G \cap M$ using
at most 
$\Delta +  \chi(M)+ 1$ colors.
\end{theorem}

\begin{proof}
We first explain how to replace $G$ by $\Delta+1$ partition matroids. 
By Vizing's theorem \cite{Vizing1964,MisraGries1992}, the edges in $G$ can be partitioned into $\Delta+1$ matchings
$S_1, \ldots, S_{\Delta+1}$, which can be computed in polynomial time. 
Now define a partition matroid $M_i= (V, \mathcal{I}_i)$
derived from $S_i$ in the following way: any vertex not matched in $S_i$ is in a singleton part,
and for each edge $(v, w) \in S_i$ there is a part consisting of exactly $v$ and $w$.
One can then observe that $\chi(G \cap M) = \chi\left(M \cap \left(\cap_{i=1}^{\Delta+1} M_i\right)\right)$,
and for all $i$, $\chi(M_i) \leq 2$. Then the statement follows directly by applying Theorem \ref{thm:main}.
\end{proof}

Note that our algorithm for a general matroid, which underlies Theorem  \ref{thm:ourstrongcoloring}, is much
simpler than the algorithm for just a partition matroid, which underlies Theorem    \ref{thm:effectivecoloring} 
(although in fairness Theorem    \ref{thm:effectivecoloring} gives a better bound when the chromatic
number of the matroid is large).

We touch on a few additional related results here. 
A graph $G = (V, E)$ is strongly $k$-colorable if for all partition matroids $M$ on ground set $V$ with chromatic number $k$, $G \cap M$ has chromatic number $k$. The strong coloring number $s\chi(G)$ of a graph $G$ is the minimum $k$ s.t. $G$ is strongly $k$-colorable. This concept was introduced independently by Alon \cite{alon_sc} and Fellows \cite{fellows_sc}.  Haxell \cite{haxell_scn_asymp} also gave a nonconstructive proof that $s\chi(G) \leq 11/4 \cdot \Delta + o(\Delta)$. The fractional strong coloring number is at most $2\Delta$ when $M$ is a partition matroid \cite{scc} and even when $M$ is an arbitrary matroid \cite{ab06}. 

\subsection{Matroid Intersection List Coloring Algorithm}\label{app:lc}

In this section, we extend our matroid intersection coloring algorithm from Section \ref{sec:alg} to list coloring. We then prove the following strengthened version of Theorem \ref{thm:main}.

\begin{theorem}\label{thm:main_lc}
There is a polynomial-time algorithm that, given a general matroid $M_1 = (X, \mathcal{I}_1)$ 
   and $k-1$ partition matroids $ M_2, \ldots, M_k$, a positive integer $m \geq 1+\sum_{i=1}^k \left(\chi(M_i) -1\right)$, and lists of positive integers $L_x$ where $\bigcup_x L_x = [m]$ and $|L_x| = 1+\sum_{i=1}^k \left(\chi(M_i) -1\right)$ for each $x \in X$, will produce a feasible coloring of the intersection $M = \cap_{i=1}^k M_i$ such that each $x \in X$ is given a color in $L_x$.
\end{theorem}

Note that plugging in $m = 1+\sum_{i=1}^k \left(\chi(M_i) -1\right)$ and therefore $L_x = [m]$ for each $x \in X$ gives Theorem \ref{thm:main} as a corollary.

The algorithm is essentially the same, but uses a \textit{pruned} Edmonds digraph $\tilde{G}$, where we prune edges corresponding to recolorings which are infeasible under the list-coloring constraints. Working with a pruned Edmonds digraph $\tilde{G}$ presents two major concerns. First, the color-chordless subgraph $H$ in $\tilde{G}$ will clearly be color-chordless in $\tilde{G}$, but we must argue that it is also color-chordless in $G$, in order to argue that the recolorings along any path $P$ can be executed simultaneously. Second, we must argue that the color-chordless subgraph construction in $\tilde{G}$ will still discover an uncolored element $u$. We resolve the first concern in Lemma \ref{lemma:list_H_paths} and the second concern in Lemma \ref{lemma:list_main_struct}.

\subsubsection{Algorithm Description}
\label{subsec:lc_intersectionalgorithmdescription}

Our Matroid Intersection List Coloring Algorithm has the following input, output, and outer loop invariant.

\textbf{Input:} An arbitrary matroid $M_1$ and partition matroids $M_i$ for $2 \leq i \leq k$ on common ground set $X$. A positive integer $m \geq \alpha + B$ for $\alpha = \chi(M_1)$ and $B = \sum_{i=2}^k \left(\chi(M_i)-1\right)$. Lists $L_x \subseteq [m]$ of positive integers where $|L_x| = \alpha+B$ for each $x \in X$ and $[m] = \bigcup_{x \in X} L_x$.

\textbf{Output:} An $m$-coloring $c$ of $M = \cap_{i=1}^k M_i$ s.t. $c(x) \in L_x$ for each $x \in X$.

\textbf{Outer Loop Invariant:} After $q$ iterations, $q$ elements of the ground set are list-feasibly colored. A coloring $c$ is \textit{list-feasible} if its color classes are independent sets in all of the matroids and $c(x) \in L_x$ for each colored element $x$.

\medskip
\noindent \textbf{Description of our Matroid Intersection List Coloring Algorithm:} Iteratively update the current coloring $c$ along an arbitrary source-sink path $P$, in the color-chordless subgraph $H$ of the current \textit{pruned} Edmonds digraph $\tilde{G}$ of $M_1$, with the property that $P$ is suffix-feasible with respect to each of the partition matroids $M_2, \ldots, M_k$. 

To complete the description of our algorithm, we need to define the pruned Edmonds digraph $\tilde{G}$ of $M_1$. $H$ and $P$ are defined in the same way as before, but now with respect to $\tilde{G}$, and using $m \geq \alpha+B$ colors instead of $\alpha+B$ colors.

\begin{definition}[Pruned Edmonds Digraph]
The pruned Edmonds digraph $\tilde{G} =([m] \cup X, \tilde{A})$ is obtained from the Edmonds digraph $G$ of $M_1, c$ by iterating over the elements $x \in X$, and removing any arcs from $G$ of the form $(i, x) \in A$ for integer $i \notin L_x$ and $(y, x) \in A$ from $G$ where $c(y) \notin L_x$.
\end{definition}

We still use $G$ to refer to the Edmonds digraph of $M_1, c$.

\subsubsection{Algorithm Analysis}
In order to prove the correctness of our algorithm, we will prove the following statements. In Lemma \ref{lemma:list_H_paths}, we will show all source-sink paths in $H$ are color-chordless paths in $G$. In Lemma \ref{lemma:list_main_struct}, we will show that $H$ contains an uncolored element $u$. In Lemma \ref{lemma:list_final}, we will tie everything together and show that our algorithm maintains a list-feasible coloring $c$ in $M_1, M_2, \dots, M_k$.

\begin{lemma}\label{lemma:list_H_paths}
All source-sink paths in $H$ are color-chordless paths in $G$.
\end{lemma}

\begin{proof}
All source-sink paths in $H$ are color-chordless paths in $\tilde{G}$ by the same proof as Lemma \ref{lemma:feas_M1}. It suffices to show that all color-chordless paths in $\tilde{G}$ are color-chordless paths in $G$. Let $P$ be a color-chordless path in $\tilde{G}$. Assume for sake of contradiction that $P$ is not color-chordless in $G$, i.e. there exists an arc $(y, x) \in P$ and a node $z$ preceding $y$ in $P$ s.t. $(z, x) \in A$ and $z, y$ have the same color. Because $(y, x) \in \tilde{A}$, $c(y) \in L_x$. Because $P$ is color-chordless in $\tilde{G}$, $(z, x) \notin \tilde{A}$, so either $z \notin L_x$ or $c(z) \notin L_x$. But $z, y$ have the same color, a contradiction. Thus $P$ is color-chordless in $G$.
\end{proof}

\begin{lemma}\label{lemma:list_main_struct}
Let $(\overline{Y}, Y)$ be a source-sink separating cut in $\tilde{G}$. Then there exists a subset of arcs $A' \subseteq \delta(\overline{Y}, Y)$ such that $|A'| > B \cdot |Y|$ and no two arcs in $A'$ share both the same target node $x$ and the same source node color $j$.
\end{lemma}

\begin{proof}
We will first show that many elements in $Y$ can change colors list-feasibly in $M_1$. Each one of these color changes induces a corresponding arc in $\delta(\overline{Y}, Y)$ by the claim in Lemma \ref{lemma:main_struct}.

Let $W_j = \{x \in X : j \in L_x\}$ for $j = 1, 2, \dots, m$. By assumption, $M_1$ is an $\alpha$-colorable matroid. 
Let $T_{1, j}, T_{2, j}, \dots, T_{\alpha, j}$ be the color classes of an $\alpha$-coloring of $M_1$ restricted to $W_j$ for $j = 1, 2, \dots, m$. Let $S_1, S_2, \dots, S_m$ be the color classes of our $m$-coloring $c$. Restrict all sets to $Y$, by letting $W_j' = W_j \cap Y$, $T_{i, j}' = T_{i, j} \cap Y$, and $S_j' = S_j \cap Y$ for $i \in [\alpha]$ and $j \in [m]$.

We now loop over $j \in [m]$, and show that many elements $x \in Y \setminus S_j'$ have $S_j' \cup \{x\} \in \mathcal{I}_1$ and $j \in L_x$. For each $j \in [m]$, and for each $i \in [\alpha]$, consider $S_j', T_{i,j}'$. There exist at least $|T_{i,j}'|-|S_j'|$ elements $x \in T_{i,j}' \setminus S_j' \subseteq Y \setminus S_j'$ such that $S_j' \cup \{x\} \in \mathcal{I}_1$ via the matroid exchange property (see Definition \ref{def:matroid}). Further, because $T_{i, j}' \subseteq W_j$, all of these elements $x$ have $j \in L_x$.

Thus, for each $j \in [m]$, the number of elements $x \in Y \setminus S_j'$ such that $S_j' \cup \{x\} \in \mathcal{I}_1$ and $j \in L_x$ is at least $\sum_{i=1}^\alpha (|T_{i,j}'| - |S_j'|)$. Summing over all $j \in [m]$, the number of $j, x$ pairs with this property is at least

\begin{align*}
\sum_{j=1}^{m} \sum_{i=1}^{\alpha} \left(|T_{i, j}'| - |S_j'|\right)
&= \sum_{j=1}^m \sum_{i=1}^{\alpha} |T_{i,j}'| - \sum_{i=1}^{\alpha} \sum_{j=1}^m |S_j'| \\
&= \sum_{j=1}^m |W_j'| - \sum_{i=1}^{\alpha} |Y \setminus U| \\
&= \sum_{x \in Y} |L_x| - \sum_{i=1}^{\alpha} |Y \setminus U| \\
&= (\alpha+B)|Y| - \sum_{i=1}^{\alpha} |Y \setminus U| \\
&> (\alpha+B)|Y| - \alpha|Y| \\
&= B|Y|
\end{align*}

Via the claim in Lemma \ref{lemma:main_struct}, each one of these color changes induces a corresponding arc in $\delta(\overline{Y}, Y)$. Thus for each $j, x$ pair such that $x \in Y \setminus S_j'$, $S_j' \cup \{x\} \in \mathcal{I}_1$, and $j \in L_x$, generate a unique arc $(y, x) \in \delta(\overline{Y}, Y)$ where $y$ has color $j$, and let $A' \subseteq \delta(\overline{Y}, Y)$ be the collection of these arcs. Then $|A'| > B|Y|$ and no two arcs in $A'$ share both the same target node $x$ and source node color $j$.
\end{proof}

For the same reason as before, Lemma \ref{lemma:list_main_struct} shows that the color-chordless subgraph $H$ contains an uncolored element $u$.

\begin{lemma}\label{lemma:list_final}
$c \Delta P$ is a list-feasible coloring of $q+1$ elements of $M = \bigcap_{i=1}^k M_i$ using $m$ colors.
\end{lemma}

\begin{proof}
$P$ is a color-chordless path in $G$ by Lemma \ref{lemma:list_H_paths}, so $c \Delta P$ is a feasible coloring of $q+1$ elements of $M_1$ by Lemma \ref{lemma:color-chordless_path}. $P$ is suffix-feasible with respect to partition matroids $M_2, M_3, \dots, M_k$ via Lemma \ref{lemma:list_main_struct} and Lemma \ref{lemma:suffix-feasible}, so $c \Delta P$ is a feasible coloring of $q+1$ elements of $M_2, M_3, \dots, M_k$. Lastly, an arc $(y, x) \in \tilde{A}$ only if $y \in L_x$ or $c(y) \in L_x$. Thus $c \Delta P$ is a list-feasible coloring of $q+1$ elements of $M = \cap_{i=1}^k M_i$. 
\end{proof}

Starting with the empty coloring and applying Lemma~\ref{lemma:list_final} $n$ times proves Theorem~\ref{thm:main_lc}.

\section{Conclusion}
We give a polynomial-time algorithm for coloring the intersection of one general matroid and $k-1$ partition matroids using at most $1+\sum_{i=1}^k \left(\chi(M_i)-1\right) \leq k\max_{i \in [k]}\chi(M_i)$ colors. Via the known efficient reductions of many of the standard combinatorial matroids to partition matroids, we extend our $O(k)$-approximation result to the setting of one general matroid and $k-1$ combinatorial matroids of this type.

For $k\geq2$ general matroids, subsequent work gives
a polynomial-time algorithm using at most
$(k^2-k)\max_{i\in[k]}\chi(M_i)$ colors
\cite{arndt2026approximation}.
The main open question is whether the dependence on $k$
can be made linear, as we are able to achieve in our restricted setting.

\begin{conjecture}
\label{conj:constructive-2k-1}
There exists a polynomial-time algorithm that, given
$k$ general matroids $M_1,\ldots,M_k$, colors their
intersection using at most
$O(k)\max_{i\in[k]}\chi(M_i)$ colors.
\end{conjecture}

A challenge with extending our approach to $k$ general matroids can be seen in the case $k=2$. In particular,  there is an instance of two laminar matroids and a partial coloring $c$, where for every source-sink path $P$ in the Edmonds digraph $G$ for $M_1$, it is the
case that $c \Delta P$ is not a feasible coloring in $M_2$. 
One potential fix is to use
an augmenting tree, instead of an augmenting path, as does 
Haxell's proof of Theorem  \ref{thm:haxellstrongcoloring},
but at least one significant difficulty with this approach is
that the independence structure in general matroids is significantly
more complicated than it is for partition matroids. The authors of the subsequent work \cite{arndt2026approximation} do not directly overcome these challenges, and instead use a completely different polyhedral approach.

\section{Acknowledgments}
Stephen Arndt, Benjamin Moseley, and Kirk Pruhs were partially funded by NSF grant CCF-2550399.

\section*{Statements and Declarations}
\paragraph{Competing interests}
The authors declare no competing interests.

\paragraph{Data availability}
No datasets were generated or analyzed in this theoretical study.

\bibliography{bib}
\bibliographystyle{spmpsci}

\end{document}